\documentclass[lettersize,journal]{IEEEtran}
\usepackage{amsmath,amsfonts}
\usepackage{algorithmic}
\usepackage{algorithm}
\usepackage{array}
\usepackage[caption=false,font=normalsize,labelfont=sf,textfont=sf]{subfig}
\usepackage{textcomp}
\usepackage{stfloats}
\usepackage{url}
\usepackage{verbatim}
\usepackage{graphicx}
\usepackage{cite}
\def\BibTeX{{\rm B\kern-.05em{\sc i\kern-.025em b}\kern-.08em
    T\kern-.1667em\lower.7ex\hbox{E}\kern-.125emX}}
\usepackage{breakurl}
\usepackage{pbox}
\usepackage{amssymb}
\usepackage{amsfonts}
\usepackage{mathrsfs}
\usepackage{amsmath}
\usepackage{textcomp}
\usepackage{float}
\usepackage[dvipsnames]{xcolor}
\usepackage{multirow}
\usepackage{csquotes}
\usepackage{epstopdf} 
\usepackage{gensymb}
\usepackage{breqn}
\usepackage{siunitx}
\usepackage{hyperref} 
\newcommand{\rom}[1]{\expandafter{\romannumeral #1\relax}}
\hypersetup{colorlinks=true,colorlinks,linkcolor={blue},citecolor={blue},urlcolor={red}} 
\makeatletter
\def\endfigure{\end@float}
\def\endtable{\end@float}
\makeatother
\usepackage[dvipsnames]{xcolor}

\usepackage{subcaption}

\usepackage{booktabs, caption, nicematrix}
\usepackage[thicklines]{cancel}
\usepackage{amsthm}
\newtheoremstyle{boldtheorem}
  {\topsep}   
  {\topsep}   
  {\normalfont}
  {}          
  {\bfseries} 
  {.}         
  {.5em}      
  {}          
\theoremstyle{boldtheorem}
\usepackage[font=small,labelfont=bf]{caption}
\usepackage{enumitem}
\newlist{legal}{enumerate}{10}
\setlist[legal]{label*=\arabic*.}

\newtheorem{thm}{Theorem}
\newtheorem{rem}{Remark}%
\newtheorem{lem}{Lemma}%
\newtheorem{assum}{Assumption}%
\usepackage{xcolor}
\usepackage[thinlines]{easytable}
\usepackage{breqn}
\usepackage{multirow}

\makeatletter
\def\endthebibliography{%
  \def\@noitemerr{\@latex@warning{Empty `thebibliography' environment}}%
  \endlist
}
\makeatother
\begin{document}

\title{Enhancing the Reliability of Closed-Loop Describing Function Analysis for Reset Control
Applied to Precision Motion Systems}

\author{Xinxin Zhang, \IEEEmembership{Member, IEEE}, and S. Hassan HosseinNia, \IEEEmembership{Senior Member, IEEE} 
\thanks{Xinxin Zhang (e-mail: x.zhang-15@tudelft.nl) and S. Hassan HosseinNia  (the corresponding author, e-mail: S.H.HosseinNiaKani@tudelft.nl) are with the Department of Precision and Microsystems Engineering (PME), Delft University of Technology, Mekelweg 2, Delft, The Netherlands, 2628 CD. }}

\markboth{Journal of \LaTeX\ Class Files,~Vol.~14, No.~8, August~2021}%
{Shell \MakeLowercase{\textit{et al.}}: A Sample Article Using IEEEtran.cls for IEEE Journals}


\maketitle

\begin{abstract}
The Sinusoidal Input Describing Function (SIDF) is an effective tool for control system analysis and design, with its reliability directly impacting the performance of the designed control systems. This study enhances the reliability of SIDF analysis and the performance of closed-loop reset feedback control systems, presenting two main contributions. First, it introduces a method to identify frequency ranges where SIDF analysis becomes inaccurate. Second, these identified ranges correlate with high-magnitude, high-order harmonics that can degrade system performance. To address this, a shaped reset control strategy is proposed, which incorporates a shaping filter to tune reset actions and reduce high-order harmonics. Then, a frequency-domain design procedure of a PID shaping filter in a reset control system is outlined as a case study. The PID filter effectively reduces high-order harmonics and resolves limit-cycle issues under step inputs. Finally, simulations and experimental results on a precision motion stage validate the efficacy of the proposed shaped reset control, showing enhanced SIDF analysis accuracy, improved steady-state precision over linear and reset controllers, and elimination of limit cycles under step inputs.

\end{abstract}

\section{Introduction}
In mechatronics industries, such as semiconductor manufacturing, robotics, and optical systems, there is a continuous demand for enhanced positioning precision, speed, and stability \cite{schmidt2020design}. Linear feedback control, particularly Proportional-Integral-Derivative (PID) control, is widely used in these applications due to its effectiveness and ease of implementation. However, the limitations of linear controllers, including the ``waterbed effect" and the Bode phase-gain trade-offs \cite{chen2018beyond}, undermine their transient and steady-state performance, making it challenging for them to meet the increasing performance demands in industries.


To overcome the limitations of linear controllers and address the rising industrial demands, nonlinear control strategies have been explored \cite{wilamowski2018control}. One such advancement is reset feedback control. Reset control originates from the Clegg Integrator (CI), which is introduced by Clegg in 1958 \cite{clegg1958nonlinear}. The CI is a linear integrator that incorporates a reset mechanism \enquote{zero-crossing law,} that resets the integrator's output to zero whenever the input signal crosses zero. The Sinusoidal-Input Describing Function (SIDF) analysis \cite{guo2009frequency} reveals that the CI introduces a phase lead of 51.9°, while maintaining the gain characteristics of a linear integrator. Leveraging the gain and phase benefits, reset controllers demonstrate enhanced performance compared to linear controller, including reduced settling time, lower overshoot, and improved noise rejection in various precision motion control applications \cite{guo2009optimal, panni2014position, beerens2019reset, zhao2019overcoming, saikumar2019constant, barreiro2021reset, karbasizadeh2022continuous}. 

While reset control enhances gain-phase margins for the first-order harmonic, it also introduces high-order (beyond the first order) harmonics. To evaluate these harmonics in closed-loop reset systems, Higher-Order Sinusoidal Input Describing Function (HOSIDF) analysis is effectively used \cite{nuij2006higher, saikumar2021loop, ZHANG2024106063}. The HOSIDF analysis quantifies the magnitude and phase of the harmonics in reset systems by measuring the systems' steady-state responses to sinusoidal inputs over a frequency range \cite{nuij2006higher}. When only the first-order harmonic is considered and high-order harmonics are neglected, this is referred as the First-Order Sinusoidal Input Describing Function (FOSIDF) \cite{guo2009frequency}. In this study, both HOSIDF and FOSIDF are collectively termed SIDF analysis methods. 

Current SIDF analysis methods for closed-loop reset systems \cite{guo2009frequency, saikumar2021loop, ZHANG2024106063} assume that only two reset actions per steady-state cycle in sinusoidal-input reset systems. However, sinusoidal-input closed-loop reset systems can exhibit either two reset actions or multiple (more than two) reset actions per steady-state cycle, referred to as two-reset systems and multiple-reset systems, respectively. The two-reset assumption in SIDF analysis introduces inaccuracies when applied to multiple-reset systems, as demonstrated in Section \ref{sec: problem state}. In such cases, the validity of the SIDF analysis is compromised, and thus the reliability of the reset control system design based on this analysis is not guaranteed.

To enhance the reliability of SIDF analysis in closed-loop reset systems, the first contribution of this study presents a method for identifying multiple-reset frequency ranges where the validity of SIDF analysis is compromised. To achieve this, we first derive piecewise expressions for the steady-state trajectories of sinusoidal-input closed-loop reset control systems. Using these expressions, a method is proposed to evaluate whether the SIDF analysis satisfies the two-reset condition in closed-loop reset systems. In previous methods, verifying this condition required calculating time-domain responses across the entire frequency spectrum, a process that is computationally expensive. The new method streamlines this process, offering a more efficient approach. Experimental results from six case studies confirm the effectiveness and time-saving benefits of this method.


In addition to compromising the accuracy of closed-loop SIDF analysis, multiple-reset actions in sinusoidal-input closed-loop reset systems indicate high-magnitude high-order harmonics. These high-magnitude high-order harmonics can increase the system's sensitivity to high-frequency noise and disturbances, leading to overall performance degradation \cite{ZHANG2024106063}. 

To tackle this challenge, the second contribution of this study introduces a shaped reset control strategy that incorporates a shaping filter that enables the tuning of reset actions to reduce high-order harmonics while preserving the benefits of the first-order harmonic. A detailed design procedure for a PID-shaped reset control system is provided, aimed at reducing high-order harmonics in a CI-based reset system. Additionally, the PID-shaped reset control addresses limit-cycle issues in reset systems under step inputs. Experimental results on a precision motion stage show that by decreasing the impact of high-order harmonics, the PID-shaped reset control system enhances the reliability of SIDF analysis and improves steady-state precision, including better reference tracking accuracy, disturbance rejection, and noise suppression. Moreover, it eliminates limit-cycle problems.

The reminder of this paper is organized as follows: Section \ref{sec: preliminaries} provides background on reset control systems and and an overview of the experimental setup. Section \ref{sec: problem state} identifies two key research problems through illustrative examples, framing the study’s objectives. Section \ref{sec: main results1} introduces the first contribution: a method to distinguish between two-reset and multiple-reset actions in sinusoidal-input closed-loop reset systems, establishing two-reset conditions for SIDF analysis with validation through simulations and experiments. Section \ref{sec: New Shaped Reset Systems} presents the second contribution, proposing a shaped reset control strategy aimed at reducing high-order harmonics. Section \ref{sec: shaped_reset_control_design} outlines a design procedure for a PID shaping filter, showcased as a case study to decrease high-order harmonics and eliminate limit cycles. Section \ref{sec: Experiments Results} then provides simulation and experimental results to validate the PID-shaped reset control system’s effectiveness on a precision motion stage. Finally, Section \ref{sec: conclusion} summarizes the main findings and offers recommendations for future research directions.
\section{Preliminaries}
\label{sec: preliminaries}
This section introduces the definition of the reset feedback control system, its stability and convergence conditions, the SIDF analysis method, and the experimental setup.

\subsection{Reset Control System}

This study focuses on the frequency-domain analysis and design in closed-loop reset feedback control systems, whose block diagram is shown in Fig. \ref{fig: RCS_d_n_r_n_n}. 
\begin{figure}[htp]
	\centerline{\includegraphics[width=0.45\textwidth]{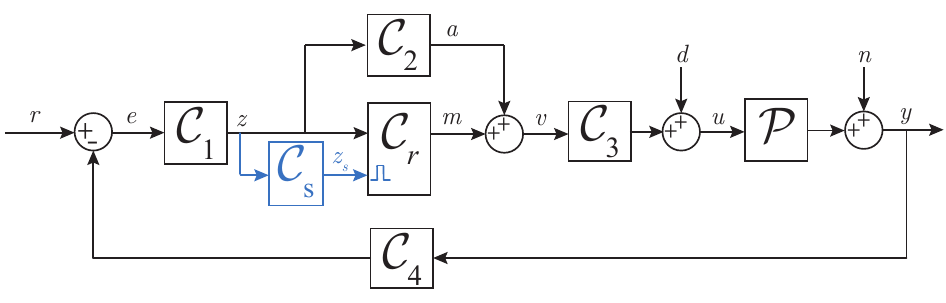}}
	\caption{Block diagram of the closed-loop reset feedback control system.}
	\label{fig: RCS_d_n_r_n_n}
\end{figure}

The system includes a nonlinear time-invariant reset controller \(\mathcal{C}_r\) with a Linear Time-Invariant (LTI) shaping filter \(\mathcal{C}_s\). Additionally, it incorporates LTI components \(\mathcal{C}_1\), \(\mathcal{C}_2\), \(\mathcal{C}_3\), \(\mathcal{C}_4\), and an LTI plant \(\mathcal{P}\). The signals \( r(t) \), \( e(t) \), \( u(t) \), and \( y(t) \) represent the reference input, error, control input, and output signals, respectively. The signals \( a(t) \) and \( v(t) \) correspond to the output of \(\mathcal{C}_2\) and the input to \(\mathcal{C}_3\), respectively.



The reset controller \(\mathcal{C}_r\) is a hybrid system. Its state-space representation, with \( z(t) \) as the input signal, \( m(t) \) as the output signal, and \( x_c(t) \in \mathbb{R}^{n_c \times 1} (n_c\in\mathbb{Z}^+ )\) as the state vector, is expressed as follows \cite{banos2012reset}:
\begin{equation} 
	\label{eq: State-Space eq of RC} 
	\mathcal{C}_r = \begin{cases}
        \dot{x}_c(t) = A_Rx_c(t) + B_Rz(t), &  t \notin J_c, \\
		x_c(t^+) = A_\rho x_c(t), & t \in J_c, \\
		m(t) = C_Rx_c(t) + D_Rz(t),\end{cases}
\end{equation} 
where \( A_R \in \mathbb{R}^{n_c \times n_c} \), \( B_R \in \mathbb{R}^{n_c \times 1} \), \( C_R \in \mathbb{R}^{1 \times n_c} \), and \( D_R \in \mathbb{R}^{1 \times 1} \) are matrices that characterize the flow dynamics of the controller. These dynamics are represented by the Base-Linear Controller (BLC), defined as:
\begin{equation}
\label{eq: RL}
C_{l}(\omega) =  C_R(j\omega I-A_R)^{-1}B_R+D_R, \ j=\sqrt{-1} , \omega\in\mathbb{R}^+.
\end{equation}
By substituting the reset controller \(\mathcal{C}_r\) with its base-linear counterpart \(\mathcal{C}_l\) as defined in \eqref{eq: RL}, the system depicted in Fig. \ref{fig: RCS_d_n_r_n_n} is referred to as the Base Linear System (BLS).

The reset controller \(\mathcal{C}_r\) employs the ``zero-crossing law" \cite{banos2012reset}, where reset actions are triggered when the ``reset-trigger signal" \(z_s(t)\) crosses zero. The set of reset instants \(t_i\) is defined as \( J_c = \{ t_i \mid z_s(t_i) = 0, \, i \in \mathbb{N} \} \). At reset instants \(t_i \in J_c\), the jump dynamics of the reset controller \(\mathcal{C}_r\) are governed by the reset matrix \(A_\rho \in \mathbb{R}^{n_c \times n_c}\), described as:
\begin{equation}
\label{eq: A_rho}
		A_\rho=\begin{bmatrix}
			\gamma & \\
			& I_{n_c-1}
		\end{bmatrix}, \ \gamma\in(-1,1)\in\mathbb{R}.
\end{equation}

\subsection{Stability and Convergence Conditions for Reset Control Systems}
Although stability and convergence are not the primary focus of this study, they are essential for the analysis and application of reset systems. These topics have been extensively explored in the literature \cite{banos2012reset}, and we outline the necessary assumptions to ensure stability and convergence.

The reset controller \eqref{eq: State-Space eq of RC} with an input signal \( e(t) = |E|\sin(\omega t + \angle E) \) has a globally asymptotically stable \( 2\pi/\omega \)-periodic solution if the following condition is satisfied \cite{guo2009frequency}:
\begin{equation}
\label{eq:open-loop stability}
    |\lambda (A_\rho e^{A_R\delta})|<1,\ \forall \delta \in \mathbb{R}^+.
\end{equation} 
Therefore, to guarantee there exists a steady-state solution for the SIDF analysis of the open-loop reset control system, the following assumption is introduced:
\begin{assum}
\label{assum: open-loop stable}
The reset system \eqref{eq: State-Space eq of RC} with input $e(t) = |E|\sin(\omega t + \angle E)$ meets the condition in \eqref{eq:open-loop stability}. Additionally, LTI systems \(\mathcal{C}_1\), \(\mathcal{C}_2\), \(\mathcal{C}_3\), and \(\mathcal{C}_4\) are Hurwitz.
\end{assum}

For a closed-loop reset control system, the conditions outlined in Assumption \ref{assum: stable} \cite{dastjerdi2022closed} ensure the system is uniformly exponentially convergent.
\begin{assum}
\label{assum: stable}
The initial condition of the reset controller \(\mathcal{C}_r\) is zero; There are infinitely many reset instants \(t_i\) such that \(\lim\nolimits_{i\to \infty} t_i = \infty\); The input signal is a Bohl function as defined in \cite{barabanov2001bohl}; The system does not exhibit Zenoness behavior; The \(H_\beta\) condition is satisfied, as per \cite{beker2004fundamental}.
\end{assum}
Assumption \ref{assum: stable} can be practically achieved through meticulous system design \cite{banos2012reset, saikumar2021loop}. 

\subsection{SIDF Analysis for Open-Loop Reset Control Systems}
For an open-loop reset system in Fig. \ref{fig: RCS_d_n_r_n_n} with the input \( e(t) = |E|\sin(\omega t + \angle E) \) and the output \( y(t) \), satisfying Assumption \ref{assum: open-loop stable}, let \( E(\omega) \) and \( Y_1(\omega) \) represent the Fourier transforms of the input signal \( e(t) \) and the first-order harmonic component of the output signal \( y(t) \), respectively. Using the SIDF analysis method \cite{karbasizadeh2022band}, the first-order transfer function of the open-loop reset system, \( L_1(\omega) \), is expressed as:
\begin{equation}
\label{eq: Ln}
\begin{aligned}
  L_1(\omega) &= \frac{Y_1(\omega)}{E(\omega)}= \mathcal{C}_1(\omega)(\mathcal{C}_r^1(\omega) + \mathcal{C}_2(\omega)) \mathcal{C}_3(\omega) \mathcal{P}(\omega),
\end{aligned}
\end{equation}
where
\begin{equation}
 \label{eq: Hn} 
    \begin{aligned}
   \Theta_\phi(\omega) &= {-2j\omega I e^{j\angle \mathcal{C}_s(\omega)}}\Omega(\omega)/{\pi}\cdot[\omega I \cos(\angle \mathcal{C}_s(\omega))\\
    &\ \ \ \ \ \ \ \  -A_R\sin(\angle \mathcal{C}_s(\omega))](\omega ^2I+{A_R}^2)^{-1} B_R,\\
   \mathcal{C}_r^1(\omega) &= C_R(A_R-j\omega I)^{-1}\Theta_\phi(\omega) +\mathcal{C}_{bl}(\omega),\\
    \Omega(\omega) &= \Delta(\omega) - \Delta(\omega){\Delta_r}^{-1}(\omega)A_\rho\Delta(\omega),\\
   \Delta _r(\omega) &= I+A_{\rho}e^{(\frac{\pi}{\omega}A_R)},\\
   \Delta(\omega) &= I+e^{(\frac{\pi}{\omega}A_R)}.
    \end{aligned}
\end{equation}
In this study, the crossover frequency \(\omega_{BW}\) of $L_1(\omega)$ where \(|L_1(\omega_{BW})| = 0\) dB in \eqref{eq: Ln}, is defined as the bandwidth frequency of a reset control system.
\subsection{Precision Motion Stage}
This paper addresses the challenge of performance in reset feedback control systems related to SIDF analysis, which is crucial for precision motion control applications. When the reliability of the SIDF frequency response analysis for closed-loop reset systems is compromised, it results in uncertainty regarding the precision of the designed reset control system. Additionally, the frequency ranges where SIDF analysis is compromised correspond to regions with high-order harmonics. If these harmonics are not properly managed in reset systems, they can cause oscillations due to high-frequency noise. Such oscillations can degrade system precision, negatively affecting stability and overall performance.

The precision motion setup used in this study is shown in Fig. \ref{fig: spider}. The motion stage is a 3 Degree-of-Freedom (DoF) system with three masses (\(M_1\), \(M_2\), and \(M_3\)), actuated by three voice coil actuators (\(A_1\), \(A_2\), and \(A_3\)). These masses are connected to a central base (\(M_c\)) via dual leaf flexures, which provide the necessary flexibility for precise motion. The actuators are driven by a linear current-source power amplifier. The control systems for this stage are implemented on an NI compactRIO platform, which includes FPGA modules for real-time processing. The digital control utilizes the \enquote{Tustin} discretization method. Position feedback is provided by a Mercury M2000 linear encoder (labeled as \enquote{Enc}, offering a high-resolution measurement of 100 nm, with the data sampled at a rate of 10 kHz.
\vspace{-0.3cm}
\begin{figure}[htp]
	\centerline{\includegraphics[width=0.4\textwidth]{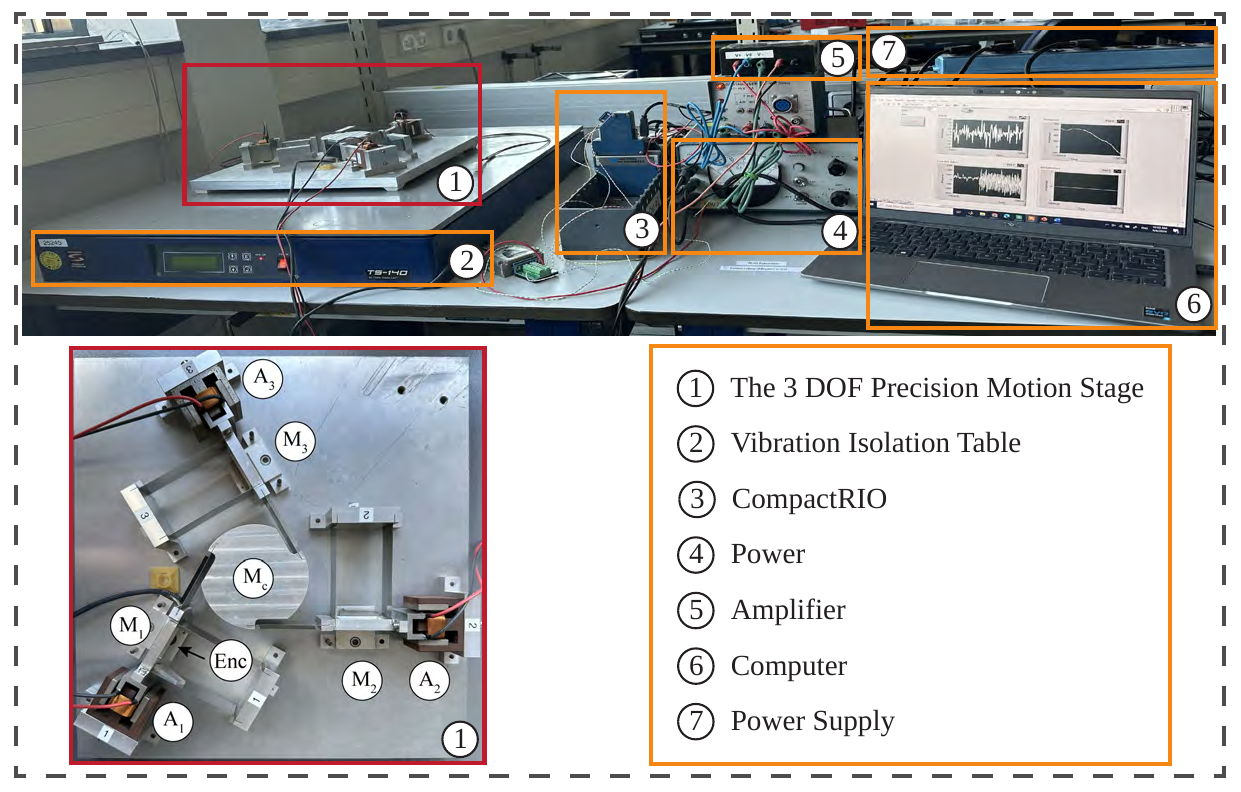}}
	\caption{Experimental precision positioning setup.}
	\label{fig: spider}
\end{figure}

\vspace{-0.3cm}
In this study, the pair of actuator \(A_1\) and mass \(M_1\) are utilized. Figure \ref{fig: spide_frf} shows the measured Frequency Response Function (FRF) from actuator \(A_1\) to mass \(M_1\). The FRF data characterize a collocated double mass-spring-damper system with high-frequency parasitic dynamics. 
\vspace{-0.3cm}
\begin{figure}[h]
	\centerline{\includegraphics[width=0.40\textwidth]{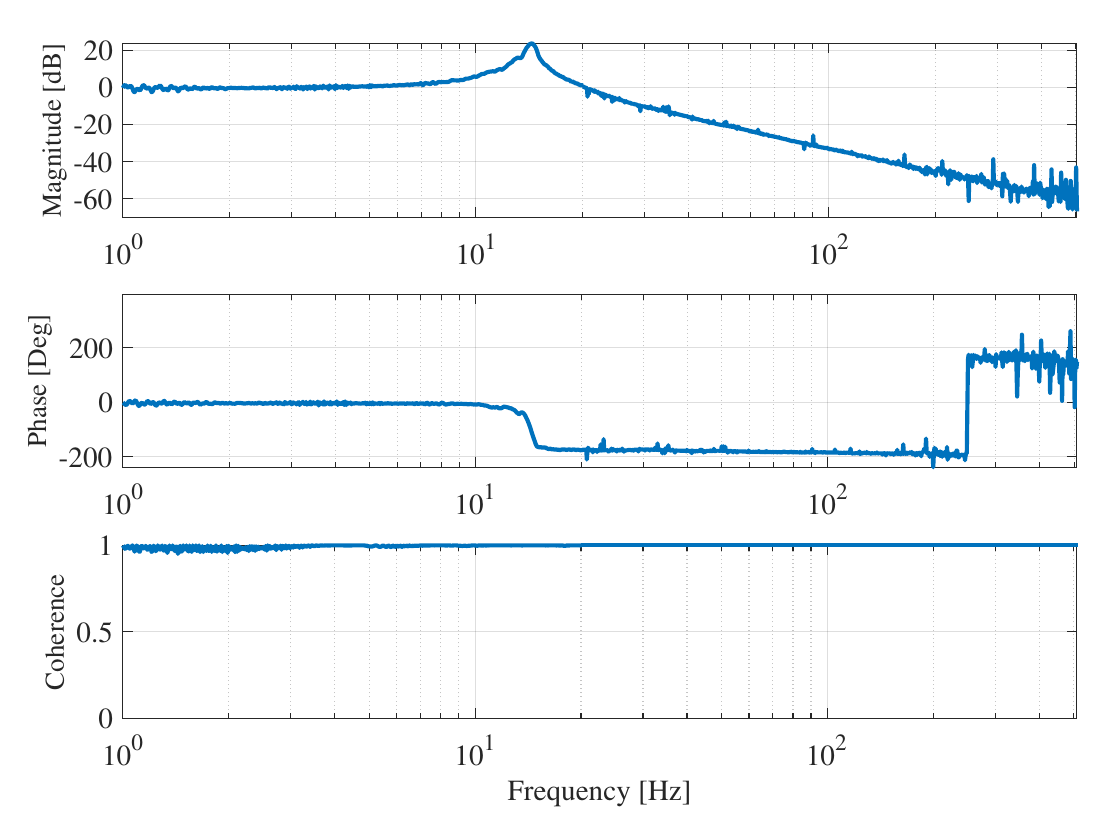}}
	\caption{FRF data from actuator $A_1$ to attached mass $M_1$.}
	\label{fig: spide_frf}
\end{figure}

\vspace{-0.3cm}
Using the system identification tools in MATLAB, the system is modeled as an LTI system represented by 
\begin{equation}
\label{eq:P(s)}
\mathcal{P}(s) = \frac{6.615 \times 10^5}{83.57s^2 + 279.4s + 5.837 \times 10^5}.
\end{equation}
This model represents the core behavior of the actuator-mass system and is used for the design and analysis of the reset control strategies discussed in the paper.
\section{Motivation and Problems Statement via Illustrative Examples} 
\label{sec: problem state}
This section outlines two research problems through examples. The first problem is that the precision of SIDF analysis is compromised by multiple-reset actions. The second problem is that these actions relate to high-order harmonics in the system, which, if large, can degrade overall system performance.

To illustrate these two problems, we design a reset control system as an example. The generalized CI is defined by \eqref{eq: State-Space eq of RC} with the matrices \( A_\rho = \gamma \in (-1, 1) \) and \((A_R, B_R, C_R, D_R) = (0, 1, 1, 0)\). A CI-based reset controller is a reset element designed using built on this generalized CI. PID controllers are widely used in mechatronics applications, and when the integrator in the PID controller is replaced by the generalized CI, the system becomes a reset PID control system. This section uses a reset PID control system to demonstrate the research problems addressed in this study.

The block diagram of the reset PID control system used in this study is depicted in Fig. \ref{Structure_PCIID}. The parameter \(\zeta\) denotes the number of integrators in the system, with this study utilizes cases where \(\zeta = 0\) and \(\zeta = 1\), referred to as Proportional-Clegg Integrator-Derivative (PCID) and PCI-PID control systems, respectively. More discussion on employing multiple integrators (\(\zeta > 1\)) is beyond the scope of this paper and can be found in \cite{karbasizadeh2022stacking}.
\vspace{-0.25cm}
\begin{figure}[htp]
    \centering
    \centerline{\includegraphics[width=0.45\textwidth]{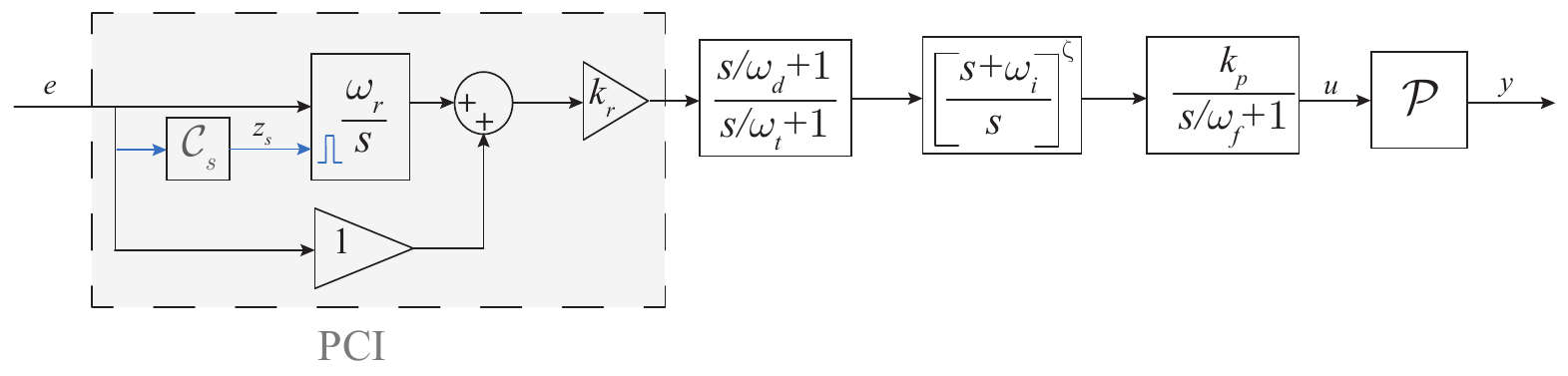}}
    \caption{Block diagram of the reset PID control system.}
	\label{Structure_PCIID}
\end{figure}

\vspace{-0.25cm}
A PCID control system is designed as the illustrative example with the following parameters: \(k_p = 17.8\), \(\omega_c = 300\pi\) [rad/s], \(\omega_r = 0.1 \omega_c\), \(k_r = 0.85\), \(\gamma = 0\), \(\omega_d = \omega_c / 3.8\), \(\omega_t = 3.8 \omega_c\), \(\omega_f = 10 \omega_c\), \(\zeta = 0\), and \(\mathcal{C}_s = 1\). A PID controller is also designed for comparison with the following parameters: \(k_p = 17.8\), \(\omega_c = 300\pi\) [rad/s], \(\omega_i = 0.084 \omega_c\), \(\omega_d = \omega_c / 3.8\), \(\omega_t = 3.8 \omega_c\), and \(\omega_f = 10 \omega_c\).

The Bode plots for the PID and the first-order harmonic of the PCID control systems are presented in Fig. \ref{PCID_Ln_final}. To ensure a fair comparison, both the PID and PCID controllers are designed to maintain the same bandwidth of 100 Hz and a phase margin of 50° with the plant \(\mathcal{P}(s)\) in \eqref{eq:P(s)}. However, the PCID controller exhibits a higher gain at frequencies below 100 Hz and a reduced gain at frequencies above 100 Hz. This design aims to enhance low-frequency tracking and disturbance rejection, and high-frequency noise suppression.
\vspace{-0.25cm}
\begin{figure}[htp]
    \centering
    \centerline{\includegraphics[width=0.38\textwidth]{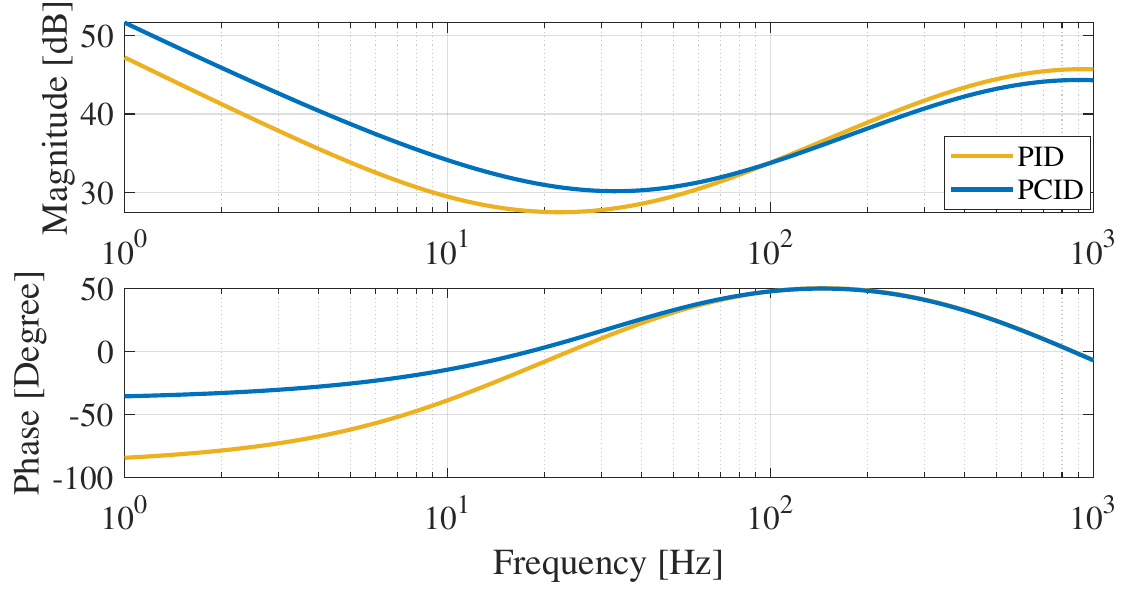}}
    \caption{Bode plots of the PID and the first-order harmonic of the PCID control systems.}
	\label{PCID_Ln_final}
\end{figure}

\vspace{-0.55cm}
\subsection{Problem 1: Multiple-Reset Actions Leading to Inaccuracies in Closed-Loop SIDF Analysis}


To evaluate the performance of closed-loop reset control systems, SIDF analysis is often employed. For closed-loop reset systems with a sinusoidal input \( r(t) = |R| \sin(\omega t) \), which satisfies Assumption \ref{assum: stable}, the sensitivity function based on SIDF analysis \cite{guo2009frequency, karbasizadeh2022band}, is defined as follows:
\begin{equation}
\label{eq: df_s1}
    \mathcal{S}(\omega) = {1}/({1 + L_1(\omega)}),
\end{equation}
where \( L_1(\omega) \) is defined in \eqref{eq: Ln}. 
\vspace{-0.3cm}
\begin{figure}[htp]
    \centering
    \centerline{\includegraphics[width=0.4\textwidth]{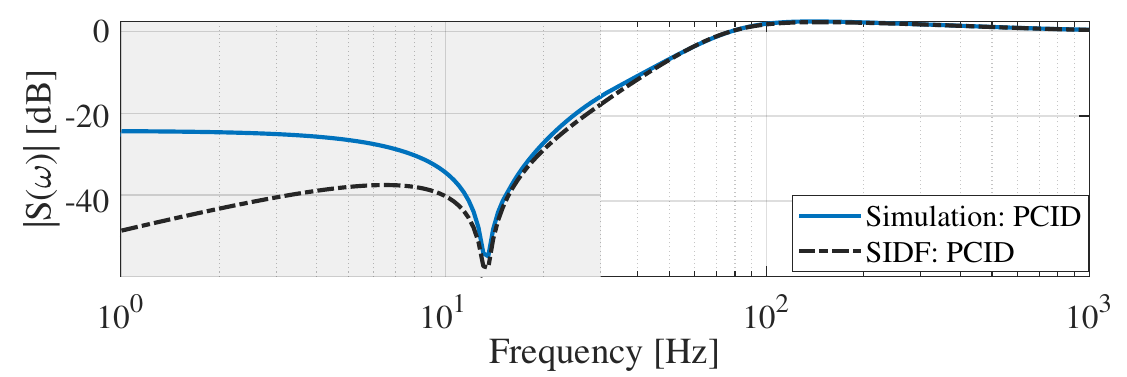}}
    \caption{The value of \( |\mathcal{S} (\omega)| \) in the PCID control system, obtained from simulation and the SIDF analysis. Multiple-reset and two-reset systems are shaded in gray and white, respectively.}
	\label{e_infty_sim_df_pcid}
\end{figure}

\vspace{-0.25cm}
The magnitude of the closed-loop sensitivity function \( |\mathcal{S}(\omega)| \) for the PCID control system, analyzed using \eqref{eq: df_s1}, is presented in Fig. \ref{e_infty_sim_df_pcid}. This analytical result is compared to the simulated value of \( |\mathcal{S}(\omega)| \), which is calculated as \( {||e||_\infty}/{||r||_\infty} \) at each frequency \( \omega \), where \( e(t) \) represents the steady-state error and \( r(t) \) denotes the input signal.

In closed-loop reset systems with a sinusoidal input \( r(t) = |R| \sin(\omega t) \), a two-reset system is defined by exactly two reset events within each \( 2\pi/\omega \) steady-state cycle, whereas a multiple-reset system has more than two reset events per cycle. In Fig. \ref{e_infty_sim_df_pcid}, the region associated with multiple-reset systems is shaded in gray, where notable discrepancies between SIDF analysis and simulation results are observed. These differences arise because the two-reset assumption in the SIDF analysis, does not hold in systems exhibiting multiple-reset actions.

Hence, to ensure the reliability of the SIDF analysis for closed-loop reset systems, it is crucial to establish a two-reset condition. The first contribution of this study in Section \ref{sec: main results1} addresses this issue. Consider a closed-loop reset system subjected to a sinusoidal input signal defined by \( r(t) = |R| \sin(\omega t) \), where \( |R| \) represents the amplitude and \( \omega \) denotes the frequency. As \( \omega \) sweeps through the operational frequency range, the proposed method identifies the frequency ranges where multiple-reset actions occur and thus deviations occur in the SIDF analysis, as illustrated by the gray area in Fig. \ref{e_infty_sim_df_pcid}.



\vspace{-0.3cm}
\subsection{Problem 2: Multiple-Reset Actions Relating to High-Order Harmonics that Need to be Reduced}
In addition to introducing imprecision in SIDF analysis, multiple-reset actions in a sinusoidal-input closed-loop reset system indicate high-magnitude high-order (beyond first-order) harmonics \cite{saikumar2021loop, ZHANG2024106063}. These large high-order harmonics can adversely affect overall system performance. This detrimental effect is further illustrated in Fig. \ref{fig: mrcs_example_pcid}.

The steady-state errors of the closed-loop PCID and PID control systems, subjected to sinusoidal reference inputs \( r(t) = \sin(20\pi t) \) (10 Hz) and \( r(t) = \sin(100\pi t) \) (50 Hz), are illustrated in Fig. \ref{fig: mrcs_example_pcid}. Additionally, the corresponding Power Spectral Density (PSD) plots are presented. To facilitate a clearer comparison, the magnitude of the first-order harmonic of the steady-state error in the PCID control system is normalized to 1, with the same scaling factor applied to the PID control system for fair comparison.


At an input frequency of 50 Hz, as shown in Fig. \ref{fig: mrcs_example_pcid}(b$_2$), the first-order harmonic component dominates, and the magnitudes of the high-order harmonics are relatively small. In this scenario, as illustrated in Fig. \ref{fig: mrcs_example_pcid}(a$_2$), the PCID control system demonstrates two-reset actions and a lower steady-state error compared to the PID control system. 

In contrast, at an input frequency of 10 Hz, the error signal exhibits multiple reset instants as shown in Fig. \ref{fig: mrcs_example_pcid}(a$_1$), which are associated with the presence of high-magnitude, high-order harmonics in Fig. \ref{fig: mrcs_example_pcid}(b$_1$). These high-magnitude high-order harmonics diminish the benefits of the first-order harmonic in the PCID control system, leading to a larger steady-state error compared to the linear PID control system in Fig. \ref{fig: mrcs_example_pcid}(a$_1$).
\vspace{-0.25cm}
\begin{figure}[htp]
      \centerline{\includegraphics[width=0.49\textwidth]{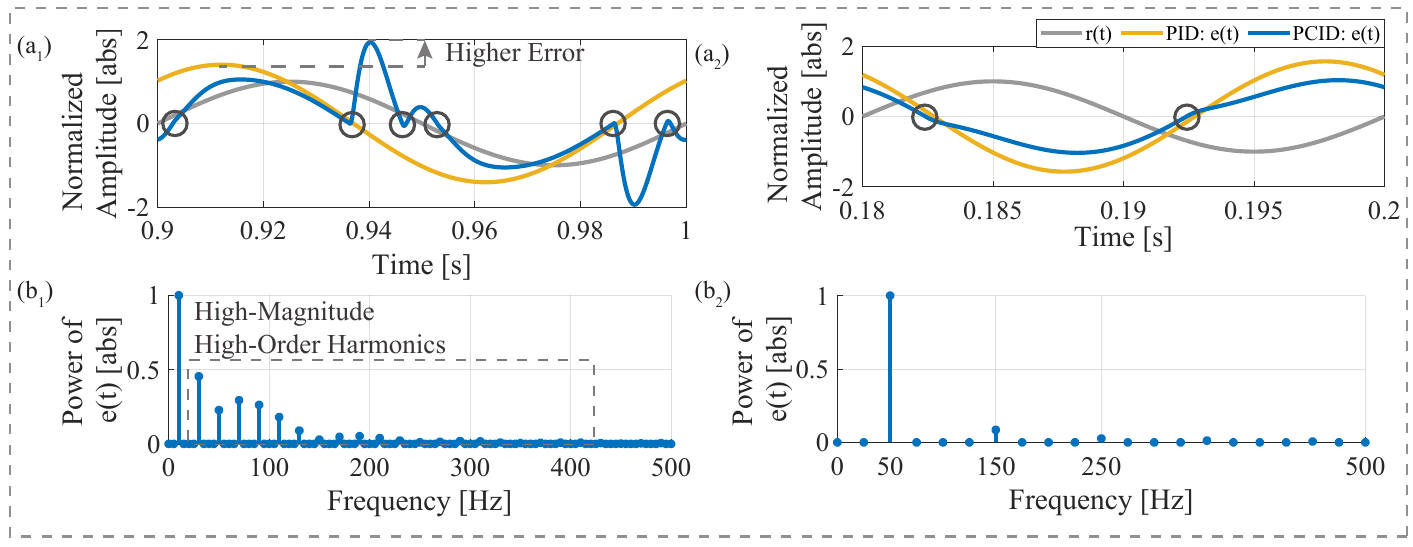}}
\caption{Steady-state errors \(e(t)\) for the PID and PCID systems under two input signals: (a$_1$) \(r(t) = \sin(20\pi t)\) and (a$_2$) \(r(t) = \sin(100\pi t)\). The gray circles mark the reset instants per cycle. Panels (b$_1$) and (b$_2$) display the PSD plots for the errors \(e(t)\) in (a$_1$) and (a$_2$), respectively.}
\label{fig: mrcs_example_pcid}
\end{figure}

\vspace{-0.25cm}
Thus, this work introduces a shaped reset control strategy to address the adverse effects of high-magnitude high-order harmonics, as detailed in Section \ref{sec: New Shaped Reset Systems}. 

Note that though the practical applications extend beyond sinusoidal-input systems, the sinusoidal-input analysis serves as an effective tool for investigating the frequency-domain harmonic characteristics within these reset control systems.

\section{Identifying Two-Reset Conditions in SIDF Analysis and Experimental Validation}
\label{sec: main results1}
In this section, first, Lemma \ref{lem: piece-wise} presents the piecewise expressions of steady-state trajectories in sinusoidal-input closed-loop reset systems. Then, building on these expressions, Theorem \ref{thm: Delta} introduces a method to identify frequency ranges of multiple-reset and two-reset actions in sinusoidal-input closed-loop reset systems. Finally, simulations and experiments validate the effectiveness of the approach in Theorem \ref{thm: Delta}.

\vspace{-0.5cm}
\subsection{Piecewise Expressions for Steady-State Trajectories in Sinusoidal-Input Closed-Loop Reset Control Systems}
\label{subsec: tool}

Consider a closed-loop reset system with a sinusoidal input \( r(t) = |R| \sin(\omega t) \) that satisfies Assumption \ref{assum: stable}. In order to conduct steady-state analysis, it is crucial to establish a reference point for one steady-state cycle. This reference point \( t_0 = 0 \) is defined at the time instant where \( r(t_0) = 0 \) and \( \dot{r}(t_0) > 0 \). 

Lemma \ref{lem: piece-wise} provides a piecewise expression of steady-state trajectories in sinusoidal-input closed-loop reset systems.
\begin{lem}
\label{lem: piece-wise}
Consider a closed-loop reset control system as shown in Fig. \ref{fig: RCS_d_n_r_n_n}, with a sinusoidal reference input \( r(t) = |R| \sin (\omega t) \), and satisfying Assumptions \ref{assum: stable}. Within one steady-state period \((0, 2\pi/\omega]\), the reset instant \( t_i \), at which \( z_s(t_i) = 0 \), divides the system trajectories into piecewise functions. Let \(x_i(t)\), \(z_i(t)\), and \(z_s^i(t)\) denote the state, reset input, and reset-triggered signal, within the intervals \((t_{i-1}, t_i]\), where \( i \in \mathbb{Z}^+ \), respectively. They are expressed as follows: 
\begin{equation}
\label{eq: xi, zi, zsi}
\begin{aligned} 
x_{i+1}(t) &= x_{i}(t) - h_{s}(t-t_i)x_i(t_i),\\       
z_{i+1}(t) &= z_{i}(t) - h_{\alpha}(t-t_i)x_i(t_i),\\
z_s^{i+1}(t) &= z_s^i(t) - h_{\beta}(t-t_1)x_i(t_i), 
\end{aligned}
\end{equation}
where
\begin{equation}
\label{eq: piece_2}
\begin{aligned}
h_{s}(t) &= \mathscr{F}^{-1}[\mathcal{T}_{s}(\omega)],\\    
h_{\alpha}(t) &= \mathscr{F}^{-1}[\mathcal{T}_{\alpha}(\omega)],\\
\mathcal{T}_{\alpha}(\omega) &=\mathcal{C}_{\sigma}(\omega)C_R\mathcal{T}_{s}(\omega),\\   
\mathcal{S}_{bl}(\omega) &= {1}/({1+\mathcal{L}_{bl}(\omega)}),\\
h_{\beta}(t) &= \mathscr{F}^{-1}[\mathcal{C}_s(\omega)\mathcal{T}_{\alpha}(\omega)],\\
\mathcal{C}_\sigma(\omega) &=\mathcal{C}_{3}(\omega)\mathcal{P}(\omega)\mathcal{C}_{4}(\omega)\mathcal{C}_{1}(\omega),\\
\mathcal{T}_{s}(\omega) &=\mathcal{S}_{bl}(\omega)(j\omega I-A_R)^{-1} (A_\rho-I),\\
\mathcal{L}_{bl}(\omega)&=\mathcal{C}_{1}(\omega)(\mathcal{C}_l(\omega)+\mathcal{C}_2(\omega))\mathcal{C}_{3}(\omega)\mathcal{P}(\omega)\mathcal{C}_{4}(\omega).
    \end{aligned}
\end{equation}
\end{lem}
\begin{proof}
    The proof is provided in \ref{appendix: proof for Lemma 1}.
\end{proof}

\subsection{Identifying Multiple-Reset and Two-Reset Actions in Sinusoidal-Input Closed-Loop Reset Control Systems}
Consider a closed-loop reset system with a sinusoidal input \( r(t) = |R| \sin(\omega t) \) that satisfies Assumption \ref{assum: stable}. Let \( t_1 \) denote the first reset instant within a single steady-state cycle. According to \eqref{eq: State-Space eq of RC}, during the interval \((0, t_1)\), the system does not undergo any reset actions, and its dynamics are determined by its BLS. However, before reaching steady-state responses, the system experiences transient trajectories. Although transient responses do not influence the steady-state trajectories in the BLS, the reset actions in the reset systems that occur during these transient states result in different initial conditions, leading to distinct trajectories compared to the BLS during the steady-state interval \((0, t_1)\). In practical reset system designs, these transient effects are often addressed through the feed-forward control and high-bandwidth feedback loops. To streamline the analysis and obviate the need for transient response calculations, the following assumption is introduced.
\begin{assum}
\label{assum: t1}
The closed-loop reset control system depicted in Fig. \ref{fig: RCS_d_n_r_n_n}, under the sinusoidal reference input \(r(t) = |R|\sin(\omega t)\) and satisfying Assumption \ref{assum: stable}, follows the same steady-state trajectory as its Base Linear System (BLS) during the time interval \((0, t_1)\), where \(t_1\) represents the first reset instant of this system within one steady-state cycle.
\end{assum}

While Assumption \ref{assum: t1} may introduce deviations in the steady-state analysis, these deviations will be evaluated through case studies in Section \ref{result 1}. 

Then, Theorem \ref{thm: Delta} and Remark \ref{rem: 2rcs} delineate the condition for ensuring the two-reset assumption in the SIDF analysis methods (\cite{guo2009frequency, saikumar2021loop, ZHANG2024106063}) for closed-loop reset control systems.
\begin{thm}
\label{thm: Delta}
Consider a closed-loop reset control system illustrated in Fig. \ref{fig: RCS_d_n_r_n_n} with a sinusoidal reference input \( r(t) = |R| \sin (\omega t) \), satisfying Assumptions \ref{assum: stable} and \ref{assum: t1}. The system is a multiple-reset system if there exists at least one time instant \( t_\delta \in (0, t_m) \), such that:
   \begin{equation}
   \label{eq: mrcs_cond2}
         \Delta(t_\delta)= |\mathcal{S}_{ls}(\omega)|\sin(\omega t_\delta)+ h_{\beta}(t_\delta) \Theta_s(\omega ) =0,
   \end{equation}
where $h_{\beta}(t)$ is given in \eqref{eq: piece_2} and
\begin{equation}
\label{eq: ht}
\begin{aligned}
 t_m &= \angle\mathcal{S}_{ls}(\omega)/{\omega} + \pi/\omega \cdot \text{sign}(\mathcal{S}_{ls}(\omega),\\ 
\Theta_s(\omega) &= |\Theta_{bl}(\omega)|\sin(\angle \mathcal{S}_{ls}(\omega) - \angle \Theta_{bl}(\omega)),\\
\Theta_{bl}(\omega)&= (j\omega I-A_R)^{-1}B_R\mathcal{C}_{1}(\omega)\mathcal{S}_{bl}(\omega),\\
\mathcal{S}_{ls}(\omega) &= \mathcal{C}_{s}(\omega)\mathcal{C}_{1}(\omega)\mathcal{S}_{bl}(\omega),\\
 \text{sign}(x) &= \begin{cases} 
      0, & \text{if } x > 0, \\
      1, & \text{if } x \leq 0.
   \end{cases}
\end{aligned}
\end{equation}

\end{thm}
\begin{proof}
    The proof is provided in \ref{appendix: proof for Theorem 1}.
\end{proof}
Theorem \ref{thm: Delta} is applicable to model-based reset control. To use it, first, the FRF data of the plant \(\mathcal{P}(s)\), is measured, and system identification methods are employed to derive the system model. Then, Theorem \ref{thm: Delta} is applied to identify the multiple-reset frequency range in sinusoidal-input closed-loop reset systems. However, if the system identification is inaccurate, the accuracy of Theorem \ref{thm: Delta} may also be compromised. Additionally, deviations may arise from Assumption \ref{assum: t1} if the transient response exhibits large impact on the steady-state behavior. These deviations will be discussed and validated through case studies in Section \ref{result 1}.

Based on Theorem \ref{thm: Delta}, Remark \ref{rem: 2rcs} establishes the two-reset condition for the SIDF analysis of closed-loop reset systems.
\begin{rem}
\label{rem: 2rcs}
The SIDF analysis for closed-loop reset systems assumes a two-reset condition. This condition holds if, for all frequencies \(\omega\) within the SIDF analysis frequency range, the criteria outlined in Theorem \ref{thm: Delta} is not met.
\end{rem}

\vspace{-0.3cm}
\subsection{Simulations and Experimental Validation of Theorem \ref{thm: Delta}}
\label{result 1}
To validate Theorem \ref{thm: Delta}, six CI-based reset controllers are designed and implemented on the precision motion system \(\mathcal{P}(s)\) defined in \eqref{eq:P(s)} as case studies. CI-based reset control systems are chosen for this validation because they are easily implemented within the classical PID control framework, but they often encounter multiple-reset actions in SIDF analysis \cite{saikumar2021loop}. Ensuring the reliability of their SIDF analysis would facilitate their practical application. The systems are configured with the following parameters:
\begin{enumerate}
    \item Case$_1$: a PCID control system, using the same design parameters outlined in Section \ref{sec: problem state}.
    \item Case$_2$: $\mathcal{C}_r$ is built on a BLC $\mathcal{C}_l(s)={125.7}/{s}$ with $\gamma=0$, $\mathcal{C}_1(s)=\mathcal{C}_2(s)=\mathcal{C}_s(s)=\mathcal{C}_4(s)=1$, $\mathcal{C}_3(s)=40.0\cdot({s/711.1+1})({s/(8.8\times10^3)+1})\cdot{1}/({s/(2.5\times10^4)+1})$.
    \item Case$_3$: $\mathcal{C}_r$ is built on a BLC $\mathcal{C}_l(s)={125.7}/{s}$ with $\gamma=0$, $\mathcal{C}_1(s)=\mathcal{C}_2(s)=\mathcal{C}_s(s)=\mathcal{C}_4(s)=1$, $\mathcal{C}_3(s)= 25.0\cdot({s/327.7+1})({s/(4.8\times10^3)+1})\cdot{1}/({s/(1.3\times10^4)+1})$.
    \item Case$_4$: $\mathcal{C}_r$ is built on a BLC $\mathcal{C}_l(s)={47.1}/{s}$ with $\gamma=0$, $\mathcal{C}_1(s)=\mathcal{C}_2(s)=\mathcal{C}_s(s)=\mathcal{C}_4(s)=1$, $\mathcal{C}_3(s)= 24.0\cdot({s/216.6+1})/({s/(4.1\times10^3)+1})\cdot(1+{94.2}/{s})\cdot{1}/{(s/(9.4\times10^3)+1)}$.
    \item Case$_5$: $\mathcal{C}_r$ is built on a BLC $\mathcal{C}_l(s)={94.2}/{s}$ with $\gamma=0.3$, $\mathcal{C}_1(s)=\mathcal{C}_2(s)=\mathcal{C}_s(s)=\mathcal{C}_4(s)=1$, $\mathcal{C}_3(s)= 20.5\cdot{(s/196.1+1)}/{(s/(4.5\times10^3)+1)}\cdot(1+{94.2}/{s})\cdot{1}/{(s/(9.4\times10^3)+1)}$.
    \item Case$_6$: $\mathcal{C}_r$ is built on a BLC \(\mathcal{C}_l = {(30\pi)}/{s}\) with the reset value $\gamma = 0$, \(\mathcal{C}_1(s) = {1/{(s/(150\pi) + 1)}}\), \(\mathcal{C}_s(s) = {(s+1)}/{(s+2)}\), \(\mathcal{C}_2(s) = 1\), \(\mathcal{C}_3(s) = 20.5 \cdot (s/(150\pi) + 1)/(s/(3000\pi) + 1) \cdot {(s/(62.5\pi) + 1)}/{(s/(1440\pi) + 1)} \cdot {(1+15\pi/s)} \cdot {1}/{(s/(3000\pi) + 1)}\).
\end{enumerate}
All systems have been verified to be stable and convergent.

In these six case studies, multiple-reset actions occur at frequencies below a certain frequency, denoted \( f_b \) in predictions from Theorem \ref{thm: Delta} and as \( f_b' \) in simulations, with deviations \(|f_b - f_b'|\) summarized in Table \ref{tb: six_cases_results}. Both prediction and simulation methods sweep the frequency range from 1 Hz to 50 Hz with a step of 1 Hz. At each frequency, the sampling rate is set to \(10^4\). Analysis shows discrepancies between 1 and 4 Hz across the cases, primarily attributed to the exclusion of transient response effects as outlined in Assumption \ref{assum: t1}.
\vspace{-0.25cm}
\begin{table}[htp]
\caption{The Theorem \ref{thm: Delta}-predicted and simulated boundary frequencies $f_b$ and $f'_b$ that separate the two-reset and multiple-reset systems, as well the computation time in Case$_1$ to Case$_6$.}
\label{tb: six_cases_results}
\centering
    \setlength{\tabcolsep}{5pt} 
    \renewcommand{\arraystretch}{1.5} 
\fontsize{8pt}{8pt}\selectfont
\resizebox{0.9\columnwidth}{!}{
\begin{tabular}{|c|c|c|c||c|c|}
\hline
Systems & ${f_b}$ [Hz] & ${f'_b}$ [Hz] & $|f_b-f'_b|$ [Hz] & Prediction Time [s] & Simulation Time [s] \\ \hline 
Case$_1$ & {30} & {32} & {2} & 1.38 & 356.63 \\ 
Case$_2$ & {39} & {40} & {1} & 1.00 & 422.51 \\ 
Case$_3$ & {37} & {41} & {4} & 1.32 & 386.97 \\ 
Case$_4$ & {34} & {32} & {2} & 1.56 & 413.28 \\ 
Case$_5$ & {37} & {33} & {4} & 1.28 & 502.76 \\ 
Case$_6$ & {38} & {42} & {4} & 0.96 & 370.60 \\ \hline
\end{tabular}
}
\end{table}

\vspace{-0.25cm}
Despite deviations of 1–4 Hz between the simulation results and the predictions from Theorem \ref{thm: Delta}, the prediction method offers substantial time-saving benefits. Identifying multiple-reset occurrences through simulation or using the numerical method in \cite{dastjerdi2022closed} requires calculating the time response at each frequency across the entire operational frequency range via a \texttt{for} loop in MATLAB, followed by counting the reset instants per steady-state cycle. In contrast, Theorem \ref{thm: Delta} streamlines this process. Table \ref{tb: six_cases_results} presents a comparison of computation times for the prediction and simulation methods. Results show that Theorem \ref{thm: Delta} achieves a reduction in computation time by around 300-fold compared to the simulation approach.

If extreme precise identification of multiple-reset actions is needed, Theorem \ref{thm: Delta} can be utilized for initial estimation. Subsequent simulations can then focus on the predicted frequency range, ensuring both accuracy and efficiency in pinpointing multiple-reset occurrences.

To further validate Theorem \ref{thm: Delta}, Figure \ref{fig: C1_results} presents experimentally measured reset-triggered signals \(z_s(t)\) for systems Case\(_1\) and Case\(_6\) in response to a reference input of \(r(t) = 1 \times 10^{-6} \sin(2\pi f t)\) [m]. Testing was conducted at the predicted threshold frequency \(f = f_b\) Hz, as well as at \(f = f_b \pm 10\) Hz, over two steady-state cycles. The results show that at \((f_b - 10)\) Hz, the systems exhibit multiple-reset behavior, while at \((f_b + 10)\) Hz, they display two reset instants per cycle, characteristic of a two-reset system. At the predicted threshold frequency \(f_b\), the systems demonstrate 3–4 reset instants per cycle, indicating a transitional behavior between two-reset and multiple-reset categories. These observations confirm that \(f_b\) serves as a boundary frequency for distinguishing two-reset from multiple-reset actions, thereby validating Theorem \ref{thm: Delta} within a 10 Hz tolerance.
\vspace{-0.3cm}
\begin{figure}[htp]
    \centerline{\includegraphics[width=0.4\textwidth]{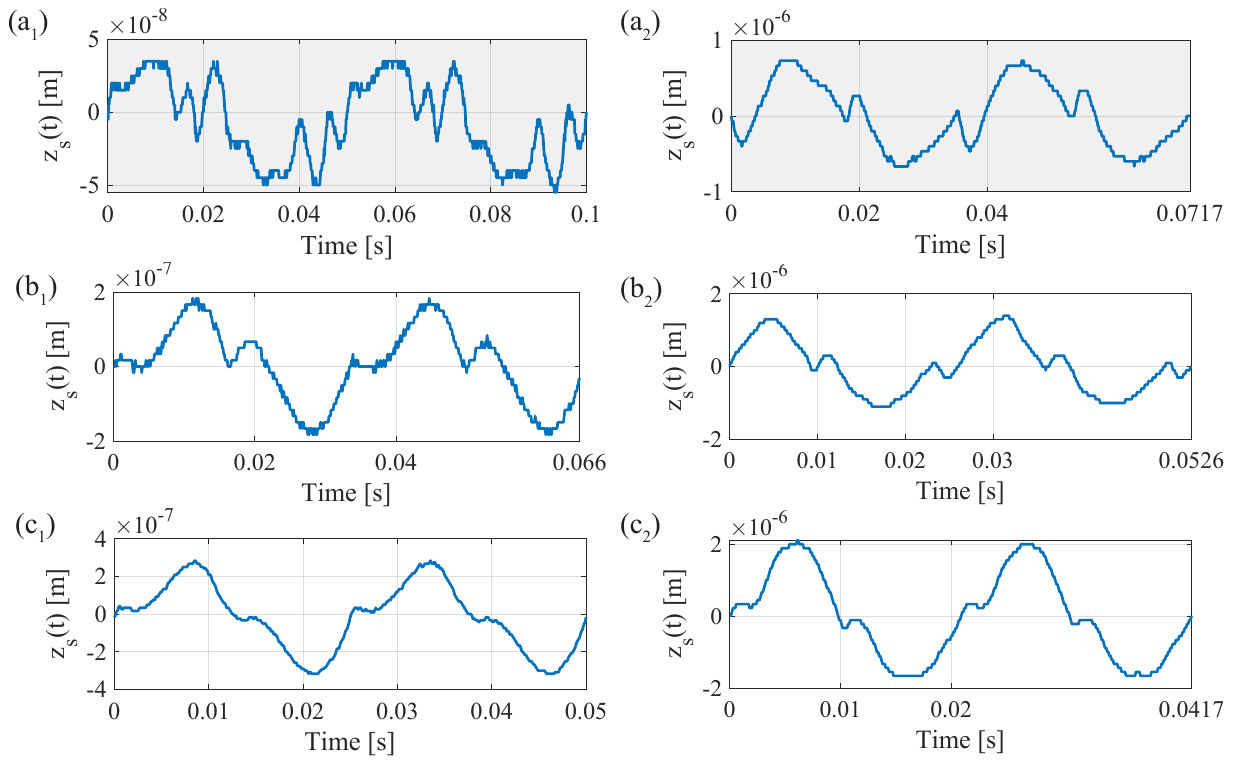}}
	\caption{Experimentally measured steady-state reset-triggered signal \( z_s(t) \) for Case\(_1\) with input frequencies of (a$_1$) 20 Hz, (b$_1$) \(f_b = 30\) Hz, and (c$_1$) 40 Hz. Steady-state reset-triggered signal \( z_s(t) \) for Case\(_6\) with input frequencies of (a$_2$) 28 Hz, (b$_2$) \(f_b = 38\) Hz, and (c$_2$) 48 Hz. Gray-shaded regions indicate multiple-reset systems.}
	\label{fig: C1_results}
\end{figure}

\vspace{-0.4cm}
\section{Analysis and Design of Shaped Reset Control Systems}
\label{sec: New Shaped Reset Systems}
Multiple-reset actions in sinusoidal-input closed-loop reset systems are indicative of high-magnitude high-order harmonics. To reduce these harmonics, this section introduces a shaped reset control strategy. First, Lemma \ref{lem: stair_step} provides an analytical decomposition of the steady-state reset-triggered signal \(z_s(t)\) in such systems into a base-linear trajectory and a nonlinear component. Building on this, Theorem \ref{thm: beta_closed_loop_reset} defines a function \(\beta_n(\omega)\), which quantifies the presence of high-order harmonics in \(z_s(t)\). This function serves as the foundation for designing a shaped reset control approach to reduce high-order harmonics.
\begin{lem}
\label{lem: stair_step}
Consider a closed-loop reset control system as shown in Fig. \ref{fig: RCS_d_n_r_n_n}, with a sinusoidal reference input \( r(t) = |R| \sin (\omega t) \) and adhering to Assumptions \ref{assum: stable} and \ref{assum: t1}. Let \(\mu\) denote the number of reset instants occurring within a half \(\pi/\omega\)-cycle. In this system, the steady-state reset-triggered signal \( z_s(t) \) is composed of two components: a base-linear element \( z_{bl}(t) \) and a nonlinear element \( z_{nl}(t) \), expressed as
\begin{equation}
\label{eq:zs,znl,zbl}
\begin{aligned}
z_s(t) & = z_{bl}(t) + z_{nl}(t),\\
z_{bl}(t) &= |R|\cdot|\mathcal{S}_{ls}(\omega)|\sin(\omega t +\angle \mathcal{S}_{ls}(\omega) ),\\
z_{nl}(t) &=  -\sum\nolimits_{n=1}^{\infty}\mathscr{F}^{-1}[\mathcal{C}_s(n\omega)\mathcal{T}_\beta(n\omega)D_s^n(\omega)],
\end{aligned}
\end{equation}
where
\begin{equation}
\label{eq:zs,znl,zbl2}
    \begin{aligned}  
\mathcal{T}_\beta(n\omega) &= \mathcal{T}_\alpha(n\omega)\cdot jn\omega,\\
D_s^n(\omega) &= \frac{2(A_\rho-I)}{n\pi}\sum\nolimits_{i=1}^{i=\mu}\mathscr{F} [x(t_i)\sin(n\omega (t-t_i))].
    \end{aligned}
\end{equation}
In \eqref{eq:zs,znl,zbl} and \eqref{eq:zs,znl,zbl2}, \(\mathcal{T}_\alpha(\omega)\) and \(\mathcal{S}_{ls}(\omega)\) are defined as in \eqref{eq: piece_2} and \eqref{eq: ht}, respectively, and \(x(t_i)\) denotes the state of the reset controller \(\mathcal{C}_r\) at the reset instant \(t_i\).
\end{lem}
\begin{proof}
 The proof is provided in \ref{appendix: proof for Lemma 2}.
\end{proof}

In the reset triggered signal $z_s(t)$, the nonlinear component \( z_{nl}(t) \) in \eqref{eq:zs,znl,zbl} can be represented as the sum of its harmonic components, expressed as:
\begin{equation}
\label{eq: z_{nl}^n,z_nl}
\begin{aligned}
z_{nl}(t) &= \sum\nolimits_{n=1}^{\infty} z_{nl}^n(t),  \\
z_{nl}^n(t) &= \sum\nolimits_{n=1}^{\infty} |Z_{nl}^n|\sin(n\omega t + \angle Z_{nl}^n), 
\end{aligned}
\end{equation}
where $|Z_{nl}^n|$ and $\angle Z_{nl}^n$ represent the magnitude and the phase of the signal $z_{nl}^n(t)$.

Let \( Z_{nl}^n(\omega) \) represent the Fourier transform of the \( n \)-th harmonic \( z_{nl}^n(t) \) within \( z_{nl}(t) \). The following theorem provides the magnitude ratio of the higher-order harmonics (\( n > 1 \)) to the first-order harmonic (\( n = 1 \)) in $z_{nl}(t) $.
\begin{thm}
\label{thm: beta_closed_loop_reset}
Consider the closed-loop reset control system depicted in Fig. \ref{fig: RCS_d_n_r_n_n}, with a sinusoidal reference input \( r(t) = |R| \sin (\omega t) \), and assume it satisfies Assumptions \ref{assum: stable} and \ref{assum: t1}. At the input frequency \( \omega \), the magnitude ratio of the higher-order harmonics (where \( n > 1 \)) to the first-order harmonic (where \( n = 1 \)) in \( z_{nl}(t) \) in \eqref{eq:zs,znl,zbl} is given by:
\begin{equation}
\label{eq: beta_def}
    \beta_n(\omega) = \frac{|Z_{nl}^n(\omega)|}{|Z_{nl}^1(\omega)|} = \frac{|\mathcal{C}_s(n\omega)\mathcal{T}_\beta(n\omega)|}{n|\mathcal{C}_s(\omega)\mathcal{T}_\beta(\omega)|}, \text{ where }n>1.
\end{equation}
\end{thm}
\begin{proof}
 The proof is provided in \ref{appendix: proof for Theorem 2}.
\end{proof}

\begin{rem}
\label{rem: beta(w) should be small}
According to \eqref{eq: beta_def}, when \(\beta_n(\omega) \to 0\), \(|Z_{nl}^n(\omega)| \ll |Z_{nl}^1(\omega)|\) holds for \(n > 1\). In this case, from \eqref{eq:zs,znl,zbl}, the reset-triggered signal \(z_s(t)\) can be approximated as \(z_s(t) \approx z_{nl}^1(t) + z_{bl}(t)\), indicating that only the first-order harmonic is present in \(z_s(t)\). This ensures the accuracy of the SIDF analysis. Conversely, the occurrence of multiple reset zero-crossings in \( z_s(t) \), which indicates multiple-reset actions within the system, is driven by high-order harmonics \( z_{nl}^n(t) \) for \( n > 1 \).
\end{rem}

However, due to the inherent nonlinearity of reset control systems, it is not feasible to completely eliminate high-order harmonics (i.e., achieve \(\beta_n(\omega) = 0\)). 

Although high-order harmonics do not always cause issues, they can lead to multiple-reset actions in sinusoidal-input closed-loop systems, compromising the accuracy of SIDF analysis and reducing the reliability of system design and performance predictions. Additionally, high-magnitude high-order harmonics increase the system's sensitivity to high-frequency disturbances and noise. To address this, we identify the multiple-reset frequency ranges as key areas where high-order harmonics should be reduced. Decreasing \(\beta_n(\omega)\) in these ranges improves the accuracy of SIDF analysis and decreases the system's sensitivity to high-frequency noise.

According to \eqref{eq:zs,znl,zbl} and \eqref{eq: beta_def}, when the base-linear component \( z_{bl}(t) \) remains constant, maintaining \(\beta_n(\omega)\) within a bound less than 1, i.e., \(\beta_n(\omega) \leq \sigma_\beta \in (0,1)\), ensures that the ratio \({|Z_{nl}^n(\omega)|}/{|Z_{nl}^1(\omega)|}\) remains within a controlled range, thereby limiting the impact of high-order harmonics.

Based on \eqref{eq: beta_n_derive}, to guide the design of a shaping filter that achieves \(\beta_n(\omega) = \sigma_\beta\), the magnitude condition for \(\mathcal{C}_s\) is given as follows:
\begin{equation}
\label{eq: desired |Cs(w)|}
|\mathcal{C}_s(\omega)| = {n\sigma_\beta}/{|\mathcal{T}_\beta(\omega)|}.   
\end{equation}
Since the reset action is independent of the magnitude of \(\mathcal{C}_s(\omega)\) \cite{banos2012reset}, the value of \(n\) does not affect the system performance. By default, \(n = 3\) is used in \eqref{eq: desired |Cs(w)|}.
Then, the following steps outline the design procedure for shaping filters in reset systems:
\begin{itemize}
    \item Step 1: Start by designing the reset control system with \(\mathcal{C}_s(\omega) = 1\), and use Theorem \ref{thm: Delta} to identify the frequency range where multiple-reset actions occur.
    \item Step 2: Then, select a value \(\sigma_\beta \in (0,1)\), and design the shaping filter \( |\mathcal{C}_s(\omega)| \) using \eqref{eq: desired |Cs(w)|} to achieve \(\beta_n(\omega) = \sigma_\beta\) within the identified multiple-reset frequency range. 
    \item Step 3: Since the introduction of \(\mathcal{C}_s(\omega)\) affects both the magnitude and phase of the first-order harmonics, adjusting other system parameters to compensate for these changes is needed in order to preserve the benefits of the first-order harmonic.
\end{itemize}

A detailed design procedure of an illustrative example following these steps is presented in Section \ref{sec: shaped_reset_control_design}.

\section{Illustrative Example: Designing a PID Shaping Filter in CI-Based Reset Control Systems}
\label{sec: shaped_reset_control_design}
This section details the analysis and design procedure for a shaped reset control systems as an illustrative example. First, Subsection \ref{subsec: PID-shaped reset control design} outlines the design process for a PID shaping filter aimed at reducing high-order harmonics in a CI-based reset control system. Next, Subsection \ref{subsec: limit_cycle_elimination} shows that this PID-shaped reset control system also addresses limit cycle issues in the step responses of reset control systems.

\vspace{-0.3cm}
\subsection{Design Procedure for a PID Shaping Filter to Reduce High-Order Harmonics in a CI-Based Reset Control System}
\label{subsec: PID-shaped reset control design}
The PCID control system, Case\(_1\), with design parameters outlined in Section \ref{sec: problem state}, is chosen as the example due to its high-order harmonic issues, as shown in Fig. \ref{fig: mrcs_example_pcid}(a$_1$). 

Following the steps outlined in Section \ref{sec: New Shaped Reset Systems}, Theorem \ref{thm: Delta} is applied to identify the multiple-reset frequency range for the PCID control system, Case\(_1\), as \((0, 30)\) Hz. The value of 30 Hz is determined by sweeping the entire frequency range with a 1 Hz step size. For improved accuracy, smaller step resolutions can be utilized. Within this identified frequency range, reducing high-order harmonics is needed.

Next, by setting \(\sigma_\beta = 0.6\) and applying equation \eqref{eq: desired |Cs(w)|}, the resulting magnitude plot of \(|\mathcal{C}_s(\omega)|\) is shown in Fig. \ref{desired_Cs_pi}. The value \(\sigma_\beta = 0.6\) is chosen based on experimental evaluations to achieve improved system performance, as demonstrated in Section \ref{sec: Experiments Results}. In practice, other values of \(\sigma_\beta \in (0, 1)\) may also be selected, depending on the specific requirements for high-order harmonic reduction in the system.

\vspace{-0.3cm}
\begin{figure}[htp]
    \centering
    \centerline{\includegraphics[width=0.8\columnwidth]{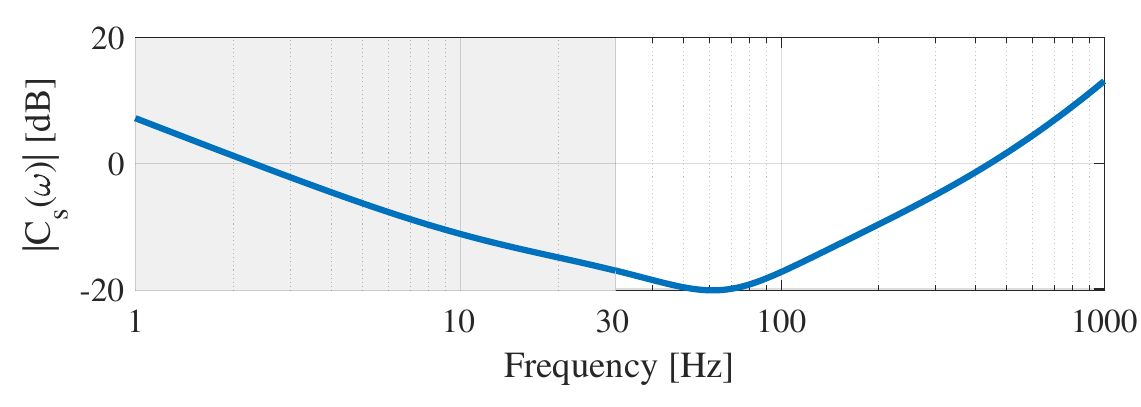}}
    \caption{The plot of $|\mathcal{C}_s(\omega)|$ meeting the condition of $\beta_n(\omega) = 0.6$ based on \eqref{eq: desired |Cs(w)|}.}
	\label{desired_Cs_pi}
\end{figure}

\vspace{-0.3cm}
From Fig. \ref{desired_Cs_pi}, the shaping filter can be simplified as the LTI PI shaping filter, given by
\begin{equation}
\label{eq: cs_pi}
\mathcal{C}_s(s) = 1+{\omega_\alpha}/{s},  
\end{equation}
where $\omega_\alpha = 2\pi \cdot 30 = 60\pi$ [rad/s].
The objective of the shaping filter design is to attenuate high-order harmonics while preserving the advantages of the first-order harmonic. However, as discussed in Appendix \ref{appendix: proof for lemma Cs_bw>0}, the PI shaping filter \eqref{eq: cs_pi} introduces a phase lag in the first-order harmonic. To maintain overall system performance, it is essential to compensate for this phase lag using a PID shaping filter, as introduced in Remark \ref{lem: angle cs_bw>0}.
\begin{rem}
\label{lem: angle cs_bw>0}
Given that the PI shaping filter \eqref{eq: cs_pi} reduces \(\beta_n(\omega)\) for frequencies \(\omega < \omega_\alpha\) in a CI-based reset control system, the PID shaping filter \(\mathcal{C}_s(s)\) further enhances performance by reducing \(\beta_n(\omega)\) for frequencies \(\omega < \omega_\alpha\) and meanwhile introducing a phase lead at the bandwidth frequency of \(\omega_{BW}\), as given by:
\begin{equation}
\label{eq: Cs}
\begin{aligned}
\mathcal{C}_s(s) = k_s\cdot \bigg(1+\frac{\omega_\alpha}{s}\bigg)\cdot\frac{s/\omega_\beta+1}{s/\omega_\eta +1}&\cdot\frac{1}{s/\omega_\psi+1},
\end{aligned}
\end{equation}
with the following conditions:
\begin{equation}
    \begin{cases}
    \omega_\beta>0 \text{ and } \omega_\eta>\omega_{BW} \text{ ensuring } \angle \mathcal{C}_s(\omega_{BW}) > 0,\\
    \omega_\psi>\omega_\eta.
    \end{cases}
\end{equation}
Since reset actions are amplitude-independent \cite{banos2012reset}, the value of \( k_s \neq 0 \in \mathbb{R} \) does not impact system performance. By default, \( k_s = 1 \). The introduction of the derivative element in \eqref{eq: Cs} may amplify high-frequency noise, potentially causing multiple-reset actions. Therefore, the low-pass filter is incorporated to attenuate high-frequency components in \( z_s(t) \). The cutoff frequency \( \omega_\psi \) for the LPF is chosen based on the characteristics of the noise present in practice.
\end{rem}
\begin{proof}
    The proof is provided in \ref{appendix: proof for lemma Cs_bw>0}.
\end{proof} 
In this study, reset systems incorporating the shaping filter from \eqref{eq: Cs} are termed PID-shaped reset control systems. 

CI-based reset systems, including Case$_1$, are built upon the generalized CI. Therefore, the phase margin introduced by the PID shaping filter in the CI-based reset system is first applied to the shaped CI, and then propagated throughout the entire system. Remark \ref{rem: phase_lead_shaped_CI} illustrates the phase lead imparted by the PID shaping filter to the shaped CI.
\begin{rem}
\label{rem: phase_lead_shaped_CI}
Based on the proof outlined in \ref{appendix: proof for lemma Cs_bw>0}, the phase lead introduced by the shaping filter in \eqref{eq: Cs} to a shaped CI is given by
\begin{equation}
\label{eq: phi_lead}
\phi_{\text{lead}} = \phi_s(\omega_{BW}) - \phi_0,    
\end{equation}
where \(\phi_s(\omega_{BW})\) denotes the phase margin of the shaped CI system when the shaping filter \(\mathcal{C}_s(s)\) in \eqref{eq: Cs} is applied, and \(\phi_0\) refers to the phase margin of the generalized CI when \(\mathcal{C}_s(s) = 1\). These phase margins are given by:
\begin{equation}
\label{eq: phi0,phi1}
\begin{aligned}
 \phi_0 &= \arctan({-\pi(1+\gamma)}/{(4(1-\gamma))}),\\
 \phi_s(\omega_{BW}) &= \arctan 
\resizebox{0.69\columnwidth}{!}{$\bigg( \frac{4(1-\gamma)\sin(\angle \mathcal{C}_s(\omega_{BW}))\cos(\angle \mathcal{C}_s(\omega_{BW}))-\pi(1+\gamma)}{4(1-\gamma)\cos^2(\angle \mathcal{C}_s(\omega_{BW}))}\bigg)$}.
\end{aligned}
\end{equation}
\end{rem}

Based on \eqref{eq: cs_pi} and Remark \ref{lem: angle cs_bw>0}, the shaping filter \(\mathcal{C}_s(s)\) \eqref{eq: Cs} for Case$_1$ is designed with the following parameters: \(\omega_\alpha = 60\pi\) rad/s, \(\omega_\beta = 1.05\cdot\omega_{BW} = 659.7\) rad/s, \(\omega_\eta = 12\cdot\omega_{BW} =  7.5\times10^3\) rad/s, \(\omega_\psi = 75\cdot\omega_{BW} = 4.7\times10^4\) rad/s, and $k_s = 213$. The Bode plot of \(\mathcal{C}_s(s)\) is shown in Fig. \ref{Cs_n_final}. At frequencies \(\omega < \omega_\alpha\), \(\mathcal{C}_s(s)\) functions as a PI controller. Additionally, its phase at the bandwidth frequency \(\omega_{BW} = 200\pi\) rad/s is \(\angle \mathcal{C}_s(\omega_{BW}) = 21^\circ\).
 
Without the shaping filter, the CI with \(\gamma = 0\) has a phase of \(\phi_0 = -38.1^\circ\) at \(\omega_{BW} = 200\pi\) rad/s, as determined using \eqref{eq: phi0,phi1}. In contrast, by applying the designed shaping filter \(\mathcal{C}_s(s)\), the phase of the shaped CI improves to \(\phi_s = -27.4^\circ\), introducing a phase lead of \(\phi_{\text{lead}} = 10.7^\circ\) in the PCID control system, as calculated from \eqref{eq: phi_lead}. To preserve the phase margin and gain properties of the first-order harmonic, the parameters are set to \(\omega_r = 141.4\) rad/s, \(\gamma = 0.13\), and \(k_r = 1.02\). Under these settings, the phase lead is \(\phi_{\text{lead}} = 0^\circ\).
\vspace{-0.3cm}
\begin{figure}[htp]
    \centering
    \centerline{\includegraphics[width=0.8\columnwidth]{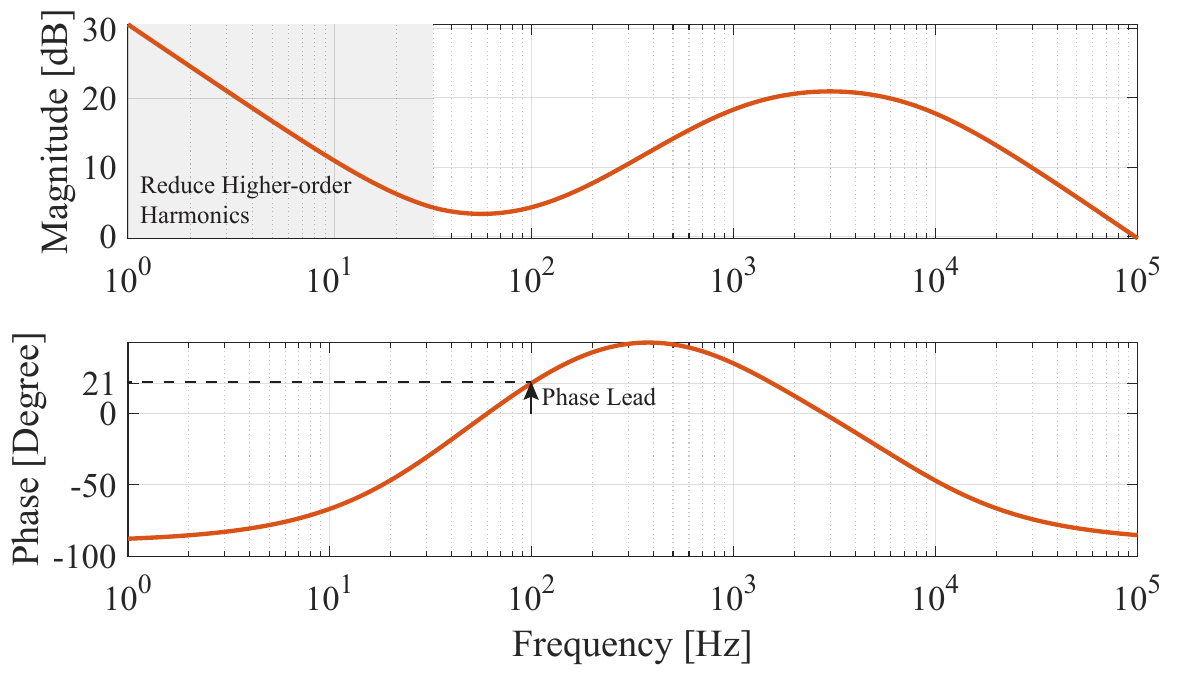}}
    \caption{Bode plot of the shaping filter \(\mathcal{C}_s(s) \).}
	\label{Cs_n_final}
\end{figure}

\vspace{-0.3cm}
Applying the designed PID shaping filter, the Bode plots for the PID, PCID, and shaped PCID controllers—showing both first- and third-order harmonics—are provided in Fig. \ref{shaped_PCID_Ln_final}. Then, Fig. \ref{shaped_PCID_sys_Ln} provides the corresponding Bode plots when these controllers are applied to the plant \(\mathcal{P}(s)\) in \eqref{eq:P(s)}. Collectively, these figures demonstrate that the shaped PCID control system reduces high-order harmonics within the frequency range of (0, 30) Hz, while preserving the gain and phase benefits of the first-order harmonic, compared to the PCID control system.
\vspace{-0.3cm}
\begin{figure}[htp]
    \centering
    \centerline{\includegraphics[width=0.8\columnwidth]{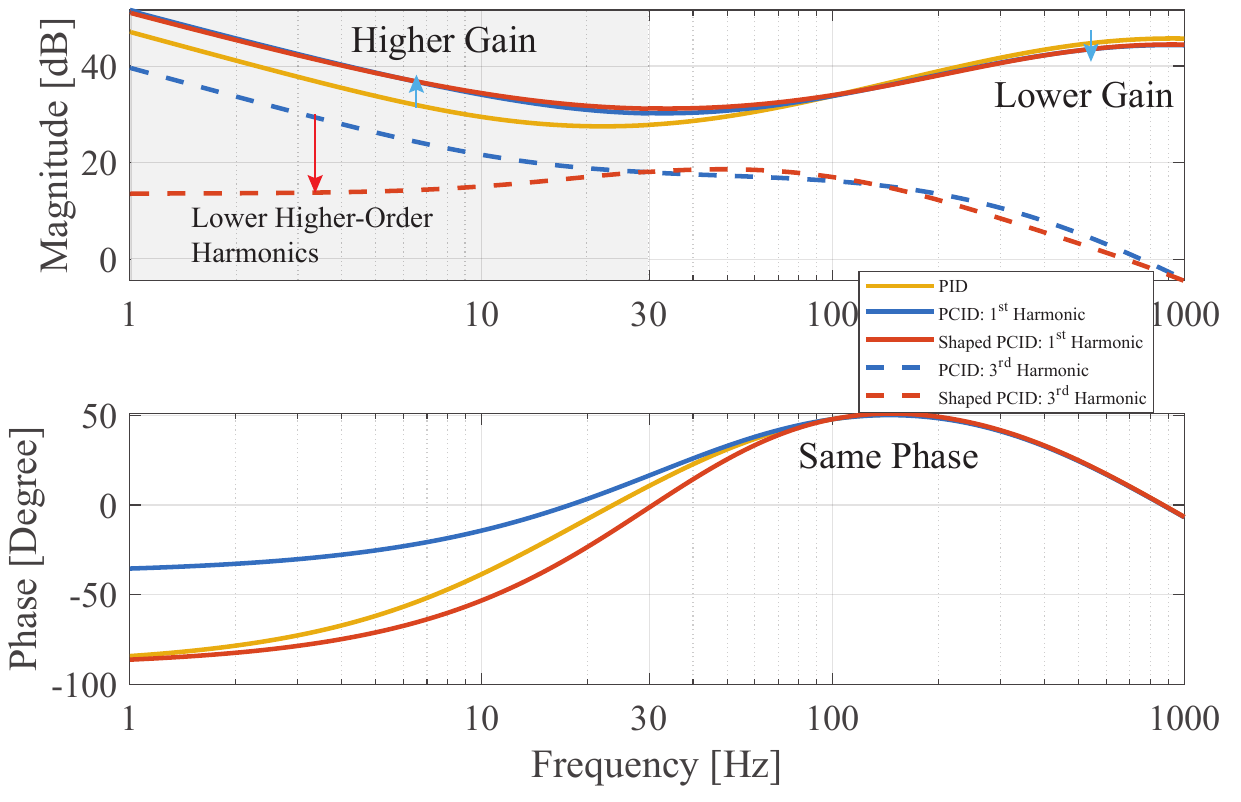}}
    \caption{Bode plots of the PID controller, the first-order and third-order harmonics in the PCID, and shaped PCID controllers. The multiple-reset region \((0,30)\) Hz identified for the PCID system using Theorem \ref{thm: Delta} is shaded in gray.}
	\label{shaped_PCID_Ln_final}
\end{figure}
\vspace{-0.3cm}
\begin{figure}[htp]
    \centering
    \centerline{\includegraphics[width=0.8\columnwidth]{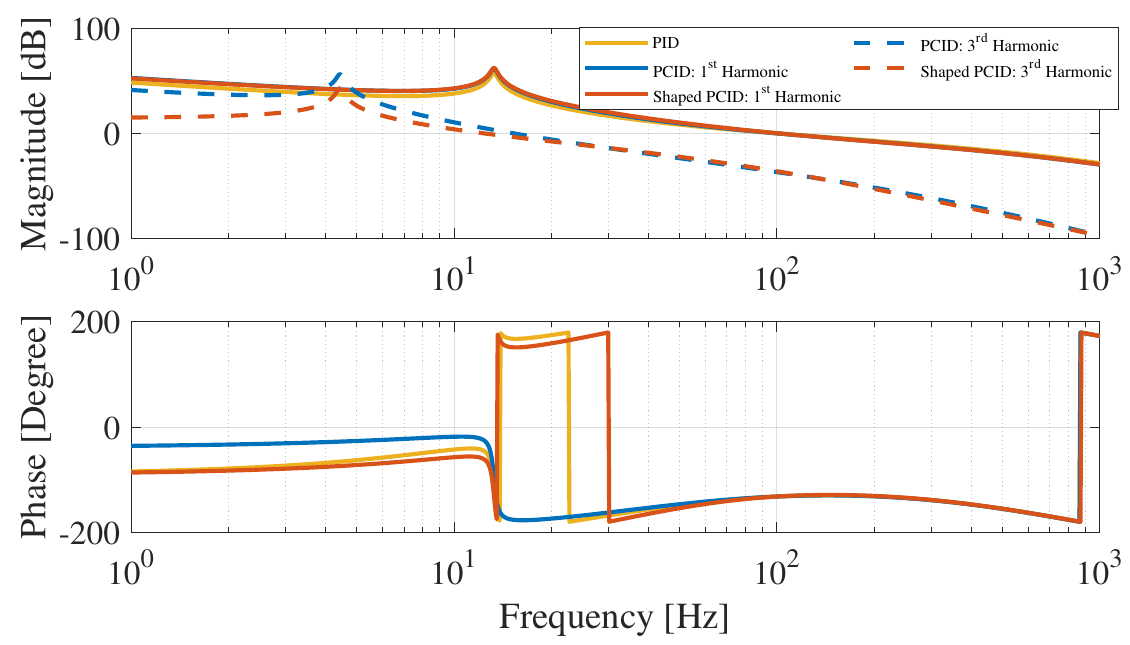}}
    \caption{Bode plots of the open-loop PID control system and the first- and third-order harmonics in the open-loop PCID and shaped PCID control systems on the precision motion stage \(\mathcal{P}(s)\).}
	\label{shaped_PCID_sys_Ln}
\end{figure}

Moreover, the plots of \(\beta_3(\omega)\) for both the closed-loop PCID and shaped PCID control systems are shown in Fig. \ref{beta_w_pcid_spcid}. The shaped PCID control system reduces \(\beta_3(\omega)\), ensuring that \(\beta_3(\omega) < 0.6\). Note that in this shaped PCID control system, the values of \(\beta_n(\omega)\) for \(n > 3\) are smaller than \(\beta_3(\omega)\) and, for clarity, are not displayed. However, they can be computed using \eqref{eq: beta_def}. 

\begin{figure}[htp]
    \centering
    \centerline{\includegraphics[width=0.8\columnwidth]{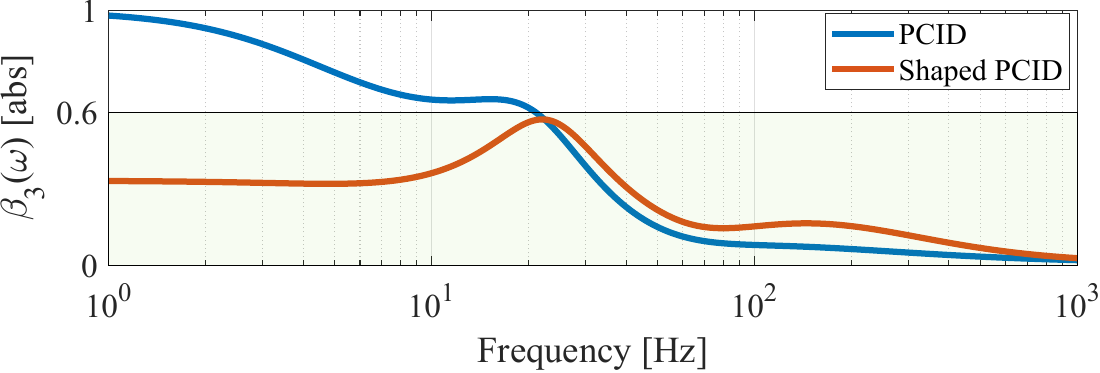}}
    \caption{The plot of \(\beta_3(\omega)\) in the closed-loop PCID and shaped PCID control systems.}
	\label{beta_w_pcid_spcid}
\end{figure}

The results shown in Figures \ref{shaped_PCID_Ln_final} to \ref{beta_w_pcid_spcid} indicate that the PID shaping filter designed in this study reduces high-order harmonics while maintaining the advantages of the first-order harmonic in the PCID system. These improvements are anticipated to enhance the accuracy of SIDF analysis and improve the steady-state precision of the PCID system. Further validation of these enhancements will be provided through simulations and experimental results in Section \ref{sec: Experiments Results}.

\subsection{Eliminating Limit Cycles Using a PID Shaping Filter}
\label{subsec: limit_cycle_elimination}
In addition to reducing high-order harmonics, the PID shaping filter also eliminates the limit cycle issues in the step responses of reset systems.

Consider a closed-loop reset control system in Fig. \ref{fig: RCS_d_n_r_n_n} subjected to an unit step input \(h(t)\), with the Laplace transform of \(h(t)\) given by \(H(s) = 1/s\). In this system, the final value of $z_s(t)$ denoted by \(\lim\limits_{t \to \infty} z_s(t)\) is given by
\begin{equation}
\label{eq: es_infty}
\begin{aligned}
\lim\limits_{t\to \infty} z_s(t) &= \lim\limits_{s\to0} s \cdot Z_s(s) = \lim\limits_{s\to0} s \, \mathcal{C}_s(s) \mathcal{C}_1(s)\mathcal{S}_{bl}(s) \cdot {1}/{s}\\
&= \lim\limits_{s\to0} \mathcal{C}_s(s)\mathcal{S}_{\alpha}(s),
\end{aligned}
\end{equation}
where
\begin{equation}
\mathcal{S}_{\alpha}(s) = \mathcal{C}_1(s)\mathcal{S}_{bl}(s).
\end{equation}
In the reset systems with the shaping filter \(\mathcal{C}_s(s) = 1\), limit cycles occur when the reset-triggered signal continues to trigger the reset actions at steady states, characterized by:
\begin{equation}
\label{eq: es_infty_cs1}
\begin{aligned}
\lim\limits_{t\to \infty} z_s(t) &= \lim\limits_{s\to0}\mathcal{S}_{\alpha}(s) = 0,
\end{aligned}
\end{equation}
while the reset controller's output $m(t)$ does not settle to a steady-state equilibrium at zero; instead, it continues to oscillate persistently around certain non-zero values as \(t \to \infty\), i.e., 
\begin{equation}
  \lim\limits_{t \to \infty} m(t)  = \lim\limits_{s \to 0} s \mathcal{C}_1(s)\mathcal{C}_l(s)\mathcal{S}_{bl}(s) 1/s = \text{constant} \neq0.
\end{equation}
%


The following content demonstrates that the PID shaping filter, as defined in \eqref{eq: Cs}, can eliminate limit cycle issues in reset systems.


The PID shaping filter \(\mathcal{C}_s(s)\) in \eqref{eq: Cs} can be expressed as:
\begin{equation}
\label{eq:pid_cs2}
\mathcal{C}_s(s) = {F(s)}/{s},
\end{equation}
where
\begin{equation}
\label{eq:F(s)}
F(s) =  k_s\cdot(s + \omega_\alpha) \cdot \frac{s / \omega_\beta + 1}{s / \omega_\eta + 1} \cdot \frac{1}{s / \omega_\psi + 1}.
\end{equation}
With the PID shaping filter in \eqref{eq:pid_cs2}, $\lim\limits_{t\to \infty} z_s(t)$ in \eqref{eq: es_infty} is written as
\begin{equation}
\label{eq: linsto0 es}
\begin{aligned}
\lim\limits_{t\to \infty} z_s(t) &= \lim\limits_{s\to0} F(s) \cdot { \mathcal{S}_{\alpha}(s)}/{s}.
\end{aligned}
\end{equation}
From \eqref{eq:F(s)}, the value of \(F(s)\) as \(s \to 0\) is given by
\begin{equation}
\label{eq: sto0F(s)}
\lim\limits_{s\to0} F(s) = k_s\cdot \omega_\alpha = \text{constant} \neq0.
\end{equation}
From \eqref{eq: es_infty_cs1}, the transfer function \( \mathcal{S}_{\alpha}(s)\) can be expressed in terms of polynomial terms, given by
\begin{equation}
\label{eq: S_bl(s)}
\begin{aligned}
\mathcal{S}_{\alpha}(s) &= \frac{n_1 s^n + n_2 s^{n-1} + \cdots + n_qs}{m_1 s^m + m_2 s^{m-1} + \cdots + m_q}, \\
& n_1, \dots, n_q, m_1, \dots, m_q \in \mathbb{R}, \quad m_q \neq 0, \, n_q \neq 0.
\end{aligned}
\end{equation}
From \eqref{eq: S_bl(s)}, we find that:
\begin{equation}
\label{eq: sS_bl(s)}
\begin{aligned}
\lim\limits_{s\to0} \frac{\mathcal{S}_{\alpha}(s)}{s} &= \lim\limits_{s\to0} \frac{n_1 s^{n-1} + n_2 s^{n-2} + \cdots + n_q}{m_1 s^m + m_2 s^{m-1} + \cdots + m_q} \\
&= {n_q}/{m_q} = \text{constant}\neq0.
\end{aligned}
\end{equation}
Combining \eqref{eq: linsto0 es}, \eqref{eq: sto0F(s)}, and \eqref{eq: sS_bl(s)}, we derive:
\begin{equation}
\label{eq: linsto0 es2}
\lim\limits_{t\to \infty} z_s(t) = \lim\limits_{s\to0} F(s) \cdot {\mathcal{S}_{\alpha}(s)}/{s} = \text{constant}\neq0.
\end{equation}
Thus, from \eqref{eq: linsto0 es2}, by applying the PID shaping filter \(\mathcal{C}_s(s)\), as specified in \eqref{eq: Cs}, the limit-cycle behaviors in reset systems are eliminated.

\section{Evaluation of the Effectiveness of the Shaped Reset Control System via Simulations and Experiments}
\label{sec: Experiments Results}

This section presents simulation and experimental results to validate the effectiveness of the PID-shaped reset system designed in Section \ref{sec: shaped_reset_control_design} in comparison to both linear and reset systems, as applied to the precision motion stage in Fig. \ref{fig: spider}.

\subsection{Simulation Results: Enhanced Steady-State Performance and Improved Accuracy of SIDF Analysis}
To evaluate the closed-loop performance of the shaped PCID control system, Fig. \ref{e_infty_final} presents the simulated \(||e||_\infty/||r||_\infty\) \eqref{eq: df_s1} for the PID, PCID, and shaped PCID systems. The shaped PCID system demonstrates the lowest \(||e||_\infty/||r||_\infty\) compared to the other two systems, indicating improved precision. This enhancement is attributed to the shaped PCID control system's superior gain properties in the first-order harmonic compared to the PID control system, while reducing high-order harmonics relative to the PCID control system, as demonstrated in Figures \ref{shaped_PCID_Ln_final} and \ref{shaped_PCID_sys_Ln}.
\begin{figure}[htp]
    \centering
    \centerline{\includegraphics[width=0.4\textwidth]{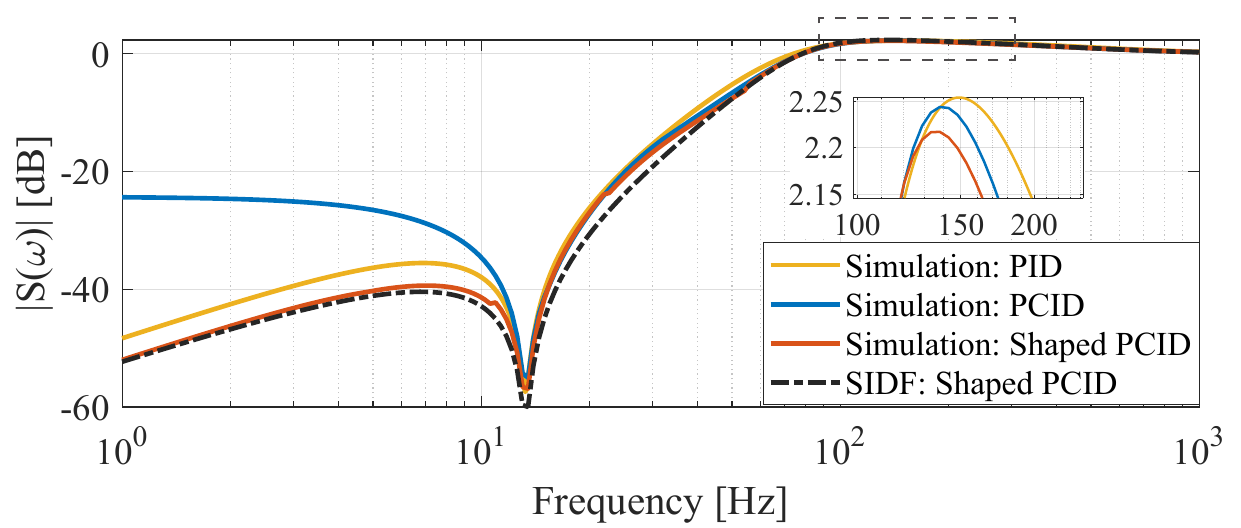}}
    \caption{Plots of simulated \(||e||_\infty/||r||_\infty\) for the PID, PCID, and shaped PCID control systems, alongside the SIDF-predicted \(||e||_\infty/||r||_\infty\) for the shaped PCID control system.}
	\label{e_infty_final}
\end{figure}

Table \ref{tb: e_max_pcid_shapedpcid} presents a quantitative comparison of \(||e||_\infty/||r||_\infty\) for the PID, PCID, and shaped PCID systems at selected frequencies: 5 Hz, 10 Hz, 30 Hz, and 200 Hz. The choice of 5, 10, and 30 Hz validates the improved precision resulting from high-order harmonics reduction in the shaped PCID system within the targeted frequency range of \( (0, 30) \) Hz. Additionally, the inclusion of 200 Hz ensures that high-frequency precision has also been attained. Across all frequencies, the shaped PCID system consistently exhibits lower steady-state errors, highlighting its effectiveness. 
\begin{table}[htp]
\caption{The \(||e||_\infty/||r||_\infty\) values for the PCID and shaped PCID systems under sinusoidal inputs at frequencies of 5, 10, 30, and 200 Hz.}
\label{tb: e_max_pcid_shapedpcid}
\centering
\renewcommand{\arraystretch}{1.5} 
\fontsize{8pt}{8pt}\selectfont
\resizebox{0.8\columnwidth}{!}{
\begin{tabular}{|c|c|c|c|c|}
\hline
\multirow{2}{*}{ Systems} & \multicolumn{4}{c|}{Input Frequencies [Hz]}\\ \cline{2-5}
&5 & 10 & 30 & 200\\ \hline 
PID  & $ 1.5\times 10^{-2}$ &  $ 1.3\times 10^{-2}$ & $ 1.7\times 10^{-1}$ &  $ 1.28$\\ 
PCID  & $ 4.8\times 10^{-2}$ &  $ 1.9\times 10^{-2}$ & $ 1.6\times 10^{-1}$ &  $ 1.26$\\ 
Shaped PCID & $ 1.0\times 10^{-2}$ & $ 8.9\times 10^{-3}$ & $ 1.5\times 10^{-1}$ & $ 1.25$ \\
Precision Improvement &  79.17\% &  53.16\% &  6.25\% &  0.79\%\\ \hline 
\end{tabular}
}
\end{table}

Another observation from Fig. \ref{e_infty_final} is that the SIDF analysis provides more reliable predictions for the shaped PCID system compared to the PCID system in Fig. \ref{e_infty_sim_df_pcid}. Define the Relative Prediction Error (RPE) of the SIDF analysis as \(\text{RPE} = ||\mathcal{S}_{\text{sim}}(\omega)| - |\mathcal{S}_{\text{sidf}}(\omega)|| / |\mathcal{S}_{\text{sidf}}(\omega)|\), where \( |\mathcal{S}_{\text{sim}}(\omega)| \) and \( |\mathcal{S}_{\text{sidf}}(\omega)| \) are obtained from simulations and SIDF analysis \eqref{eq: df_s1}, respectively. A comparison of RPE values across six frequencies, shown in Table \ref{tb: rpe_pcid_shapedpcid}, supports this observation. The improved reliability of the SIDF analysis in the shaped PCID system is attributed to the high-order harmonics reduction. However, discrepancies between SIDF predictions and simulations remain, as the SIDF considers only the first-order harmonic. To address this, \(\beta_n(\omega)\to0\) can be restricted to maintain the two-reset condition, and HOSIDF methods in \cite{saikumar2021loop, ZHANG2024106063}, can be employed for higher accuracy.

\vspace{-0.2cm}
\begin{table}[htp]
\caption{Relative Prediction Error (RPE) of the SIDF analysis for PCID and shaped PCID control systems at frequencies of 1, 10, 50, 100, 500, and 1000 Hz.}
\label{tb: rpe_pcid_shapedpcid}
\centering
\renewcommand{\arraystretch}{1.5} 
\fontsize{8pt}{8pt}\selectfont
\resizebox{0.8\columnwidth}{!}{
\begin{tabular}{|c|c|c|c|c|c|c|}
\hline
\multirow{2}{*}{ Systems} & \multicolumn{6}{c|}{Input Frequencies [Hz]}\\ \cline{2-7}
&1 & 10 & 50 & 100 & 500& 1000\\ \hline 
PCID  &  15.51 & 0.97 &  0.03 & 0.02 & 2.35$\times 10^{-3}$& 3.21$\times 10^{-3}$\\ 
Shaped PCID & 0.03 &0.18 & 0.02 & 0.01& 8.39$\times 10^{-4}$ & 3.15$\times 10^{-3}$\\
\hline 
\end{tabular}
}
\end{table}
\vspace{-0.5cm}
\subsection{Experimental Results: Improved Steady-State Tracking Precision}
Figure \ref{single_r_5_10_50_200} illustrates the experimentally measured steady-state errors for the PID, PCID, and shaped PCID systems in response to a normalized sinusoidal input signal defined as \(r(t) = 1 \times 10^{-5} \sin(2\pi f t)\) [m], with frequencies \(f = 5, 10, 30, 200\) Hz. Note that during practical experiments, the magnitudes of input signals employed for these four frequencies were different; however, for the purposes of comparison, the magnitude of all input signals have been normalized to \(1 \times 10^{-5}\) [m].

Table \ref{tb: e_max_pcid_shapedpcid_exp} presents the maximum steady-state errors at the three test frequencies for the PID, PCID, and shaped PCID systems. The results demonstrate that the shaped PCID system achieves lowest position errors among the three systems, especially at low frequency range of (0,30) Hz. Notably, at 5 Hz, the shaped PCID system improves precision by 72.56\% compared to the PCID system. 
\begin{figure}[htp]
    \centering
    \centerline{\includegraphics[width=0.8\columnwidth]{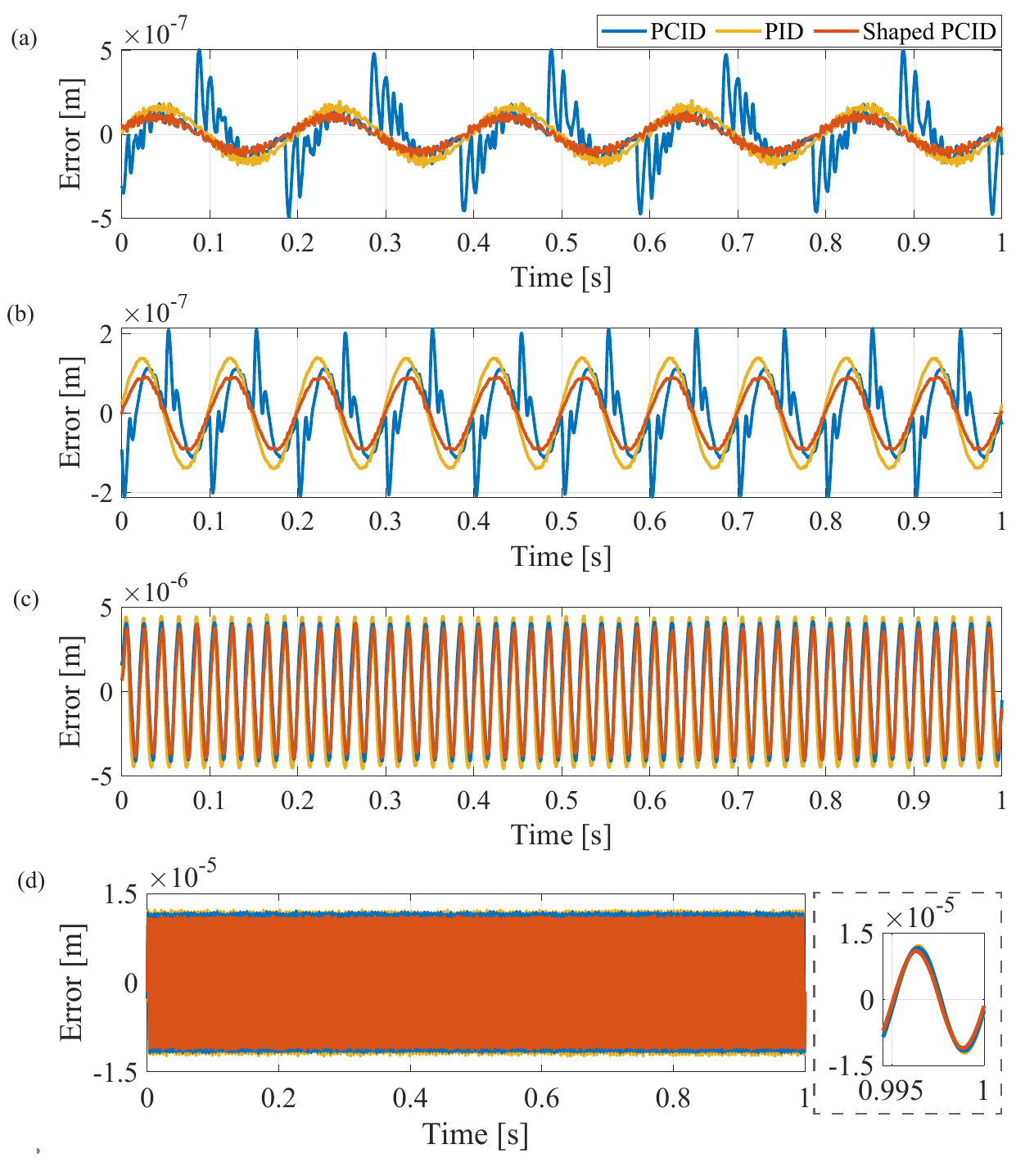}}
    \caption{Normalized experimental measured steady-state errors of the PID, PCID, and shaped PCID control systems under sinusoidal input \(r(t) = 1 \times 10^{-5} \sin(2\pi t)\) [m].}
	\label{single_r_5_10_50_200}
\end{figure}
\begin{table}[htp]
\caption{The maximum steady-state errors $||e||_\infty$ [m] in the PCID system and shaped PCID systems under single sinusoidal inputs at frequencies of 5, 10, 30, 200 Hz.}
\label{tb: e_max_pcid_shapedpcid_exp}
\centering
\renewcommand{\arraystretch}{1.5} 
\fontsize{8pt}{8pt}\selectfont
\resizebox{0.8\columnwidth}{!}{
\begin{tabular}{|c|c|c|c|c|}
\hline
\multirow{2}{*}{ Systems} & \multicolumn{4}{c|}{Input Frequency [Hz]}\\ \cline{2-5}
&5 & 10 & 30 & 200\\ \hline 
PID  & $ 2.03\times 10^{-7}$ &  $ 1.41\times 10^{-7}$ & $ 1.72\times 10^{-6}$&  $ 1.23\times 10^{-5}$\\ 
PCID  & $ 5.03\times 10^{-7}$ &  $2.16\times 10^{-7}$ & $ 1.65\times 10^{-6}$&  $ 1.21\times 10^{-5}$\\ 
Shaped PCID & $ 1.38\times 10^{-7}$ & $ 9.30\times 10^{-8}$ & $ 1.56\times 10^{-6}$ & $ 1.14\times 10^{-5}$ \\
Precision Improvement &  72.56\% &  56.94\% &  5.45\% &  7.00\%\\ \hline 
\end{tabular}
}
\end{table}


Real-world input signals are often more complex than a single sinusoid. In this subsection, the results of the single sinusoidal reference inputs serve to illustrate the steady-state performance of the three systems across varying frequencies. To comprehensively evaluate the positioning performance of the shaped reset control system, multiple inputs—including disturbances and noise—will be applied to the three systems in the next subsection.

\vspace{-0.2cm}
\subsection{Experimental Results: Improved Steady-State Tracking Precision and Disturbance and Noise Rejection}
This subsection presents the steady-state errors of three systems under multiple input conditions.

Figure \ref{dis_noise_result}(a) shows the measured steady-state errors of the three systems in response to a disturbance input signal defined as $d_1(t) = 1 \times 10^{-7}[149.3\sin(4\pi t) + 1.2\sin(16\pi t) + 11.9\sin(16\pi t) + 3.0\sin(40\pi t)] \text{ [m]}$. 

Next, a white noise input \(n(t)\) with a power bound of \(3 \times 10^{-12}\) [m] is added to the disturbance input \(d_1(t)\). The resulting steady-state errors for the three systems are presented in Fig. \ref{dis_noise_result}(b). Table \ref{tb: e_max_disturbance_noise} summarizes the maximum steady-state errors for the PID, PCID, and shaped PCID systems under these inputs. The results show that the shaped PCID system improves precision by \(80.07\%\) compared to the PCID system, effectively rejecting both the disturbance and noise.

\begin{figure}[htp]
    \centering
    \centerline{\includegraphics[width= 0.8\columnwidth]{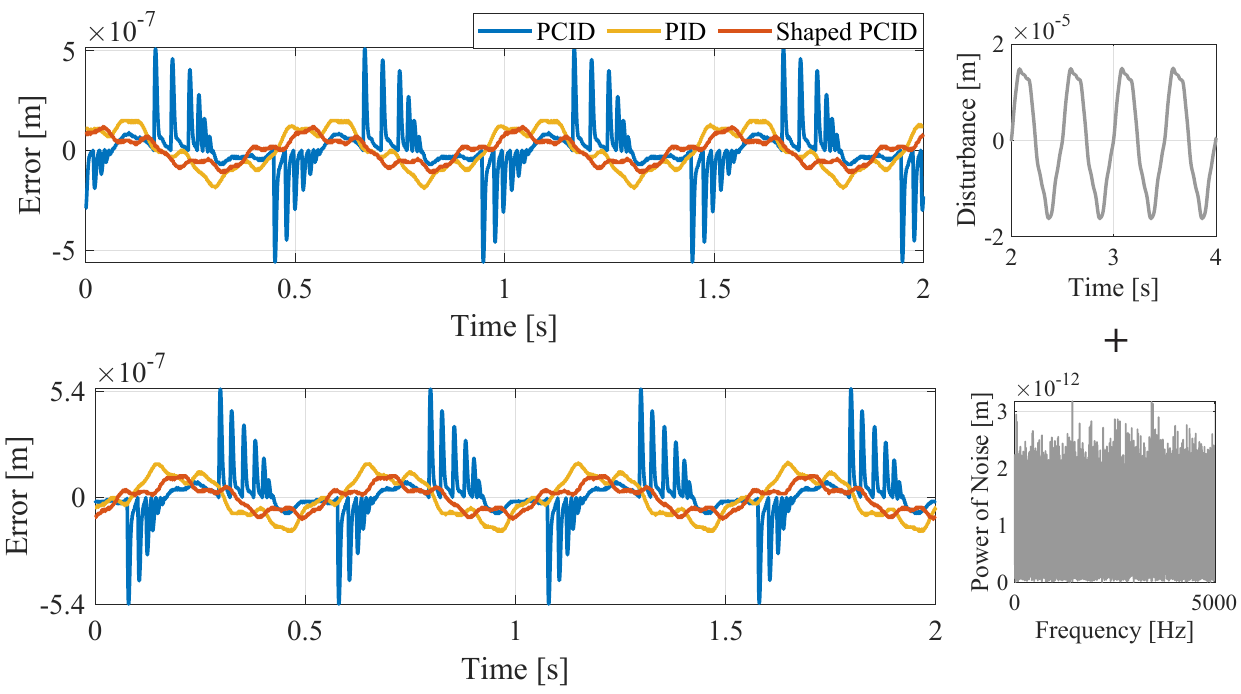}}
    \caption{Experimental measured steady-state errors of PID, PCID, and shaped PCID control systems under $d_1(t) +n(t)$.}
	\label{dis_noise_result}
\end{figure}

\vspace{-0.2cm}
\begin{table}[htp]
\caption{The maximum steady-state errors $||e||_\infty$ [m] in the PID, PCID, and shaped PCID control systems under different input signals.}
\label{tb: e_max_disturbance_noise}
\centering
\renewcommand{\arraystretch}{1.5} 
\fontsize{8pt}{8pt}\selectfont
\resizebox{0.8\columnwidth}{!}{
\begin{tabular}{|c|c|c|c|}
\hline
\multirow{2}{*}{ Systems} & \multicolumn{3}{c|}{Input Signals}\\ \cline{2-4}
  & $d_1(t)+n(t)$ & $r_2(t)+d_2(t)+n(t)$ & $r_3(t)+d_3(t)+n(t)$\\ \hline 
PID  & $ 1.80\times 10^{-7}$ &  $ 1.22\times 10^{-7}$ & $ 1.17\times 10^{-7}$\\ 
PCID  & $ 5.52\times 10^{-7}$ &  $3.63\times 10^{-7}$ & $ 2.00\times 10^{-7}$\\ 
Shaped PCID & $ 1.10\times 10^{-7}$ & $ 8.80\times 10^{-8}$ & $ 9.64\times 10^{-8}$ \\
Precision Improvement & 80.07\% &  75.78\% &  51.79\%\\ \hline 
\end{tabular}
}
\end{table}

To evaluate both the reference tracking, as well as the disturbance and noise rejection of the closed-loop shaped PCID control system, Figure \ref{r5_r10_dis_noise} compares the steady-state errors of the PID, PCID, and shaped PCID systems under multiple input signals. In Fig. \ref{r5_r10_dis_noise}(a), the inputs include a reference signal \( r_2(t) = 6 \times 10^{-6} \sin(10\pi t) \text{ [m]}, \) alongside the disturbance \( d_2(t) = 1\times 10^{-8}[49.0\sin(4\pi t) + 5.5\sin(16\pi t) + 1.1\sin(40\pi t)] \text{ [m]}\) and white noise \(n(t)\) with a power bound of \(3\times 10^{-12}\) [m]. In Fig. \ref{r5_r10_dis_noise}(b), the inputs consist of a reference signal \( r_3(t) = 6 \times 10^{-6} \sin(20\pi t) \text{ [m]},\) a disturbance \( d_3(t) = 1\times 10^{-7}[2.7\sin(10\pi t) + 3.7\sin(14\pi t) + 3.0\sin(30\pi t)] \text{ [m]},\) and the white noise noise \(n(t)\) with a power bound of \(3\times 10^{-12}\) [m]. The maximum steady-state errors for these two cases are summarized in Table \ref{tb: e_max_disturbance_noise}, indicating that the shaped PCID system improves precision by \(73.5\%\) and \(53.06\%\) in the two scenarios, respectively.

\begin{figure}[!t]
    \centering
    \centerline{\includegraphics[width= 0.8\columnwidth]{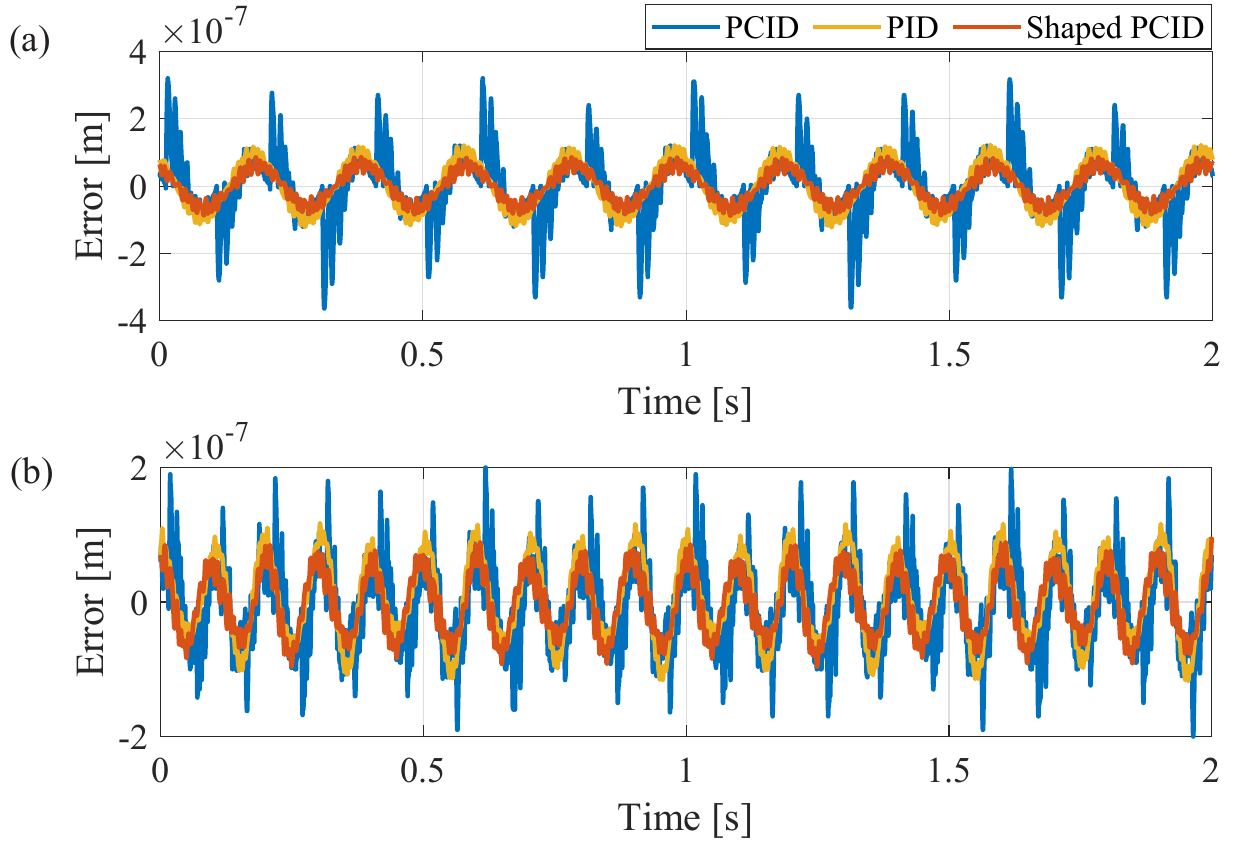}}
    \caption{Experimental measured steady-state errors of PID, PCID, and shaped PCID control systems under (a) $r_2(t) + d_2(t) +n(t)$ and (b) $r_3(t) + d_3(t) +n(t)$.}
	\label{r5_r10_dis_noise}
\end{figure}
\begin{figure}[!t]
    \centering
    \centerline{\includegraphics[width= 0.8\columnwidth]{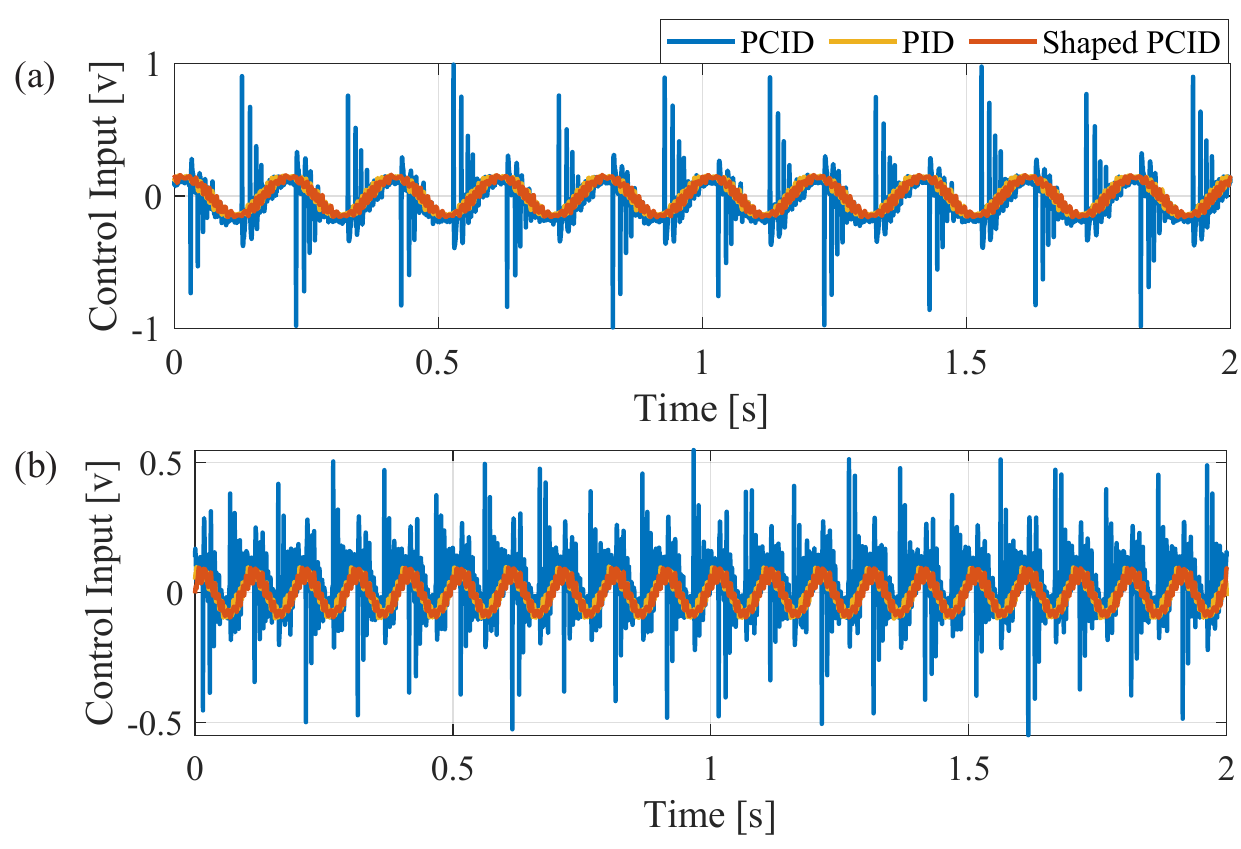}}
    \caption{Experimental measured control input of PID, PCID, and shaped PCID control systems under (a) $r_2(t) + d_2(t) +n(t)$ and (b) $r_3(t) + d_3(t) +n(t)$.}
	\label{r5_r10_dis_noise_u}
\end{figure}

Moreover, Figure \ref{r5_r10_dis_noise_u} illustrates the control inputs for these two cases, demonstrating that the shaped PCID system requires the least control input force while achieving the lowest steady-state error. Together, Figures \ref{r5_r10_dis_noise} and \ref{r5_r10_dis_noise_u} highlight the improved control efficiency of the shaped PCID system, which can be attributed to the reduction of high-order harmonics.

\vspace{-0.5cm}
\subsection{Experimental Results: Eliminated Limit Cycle}
The shaped reset system also eliminates the limit cycle issues in the step responses of reset PID systems. Current solutions for addressing the limit cycle problem include the \enquote{PI+CI} structure \cite{banos2007definition} and the PCI-PID structure in Fig. \ref{Structure_PCIID}. To provide a fair comparison of the effectiveness of five control structures—PID, PCID, shaped PCID, PI+CI D, and PCI-PID—we designed these systems with the same bandwidth of 100 Hz and phase margin of 50$\degree$ of the first-order harmonics for fair comparison. 

Figure \ref{step50_final}(a) presents the step responses of the five systems, highlighting the effectiveness of the shaped PCID, PI+CI D, and PCI-PID systems in mitigating the limit cycle issues observed in the PCID system. These systems also exhibit lower overshoot compared to the PID system. However, the PI+CI D and PCI-PID structures address limit cycle problems at the cost of reduced steady-state performance.

For example, as shown in Fig. \ref{step50_final}(b), under a sinusoidal input signal \( r(t) = 1 \times 10^{-5}\sin(20\pi t) \) [m], the steady-state errors of the PI+CI D and PCI-PID systems are larger than those of the shaped PCID system. This occurs because the PI+CI D and PCI-PID systems exhibit high-magnitude high-order harmonics at low frequencies, similar to the PCID system. In contrast, the shaped PCID system reduces these high-order harmonics, leading to improved steady-state performance.
\begin{figure}[htp]
    \centering
    \centerline{\includegraphics[width= 0.8\columnwidth]{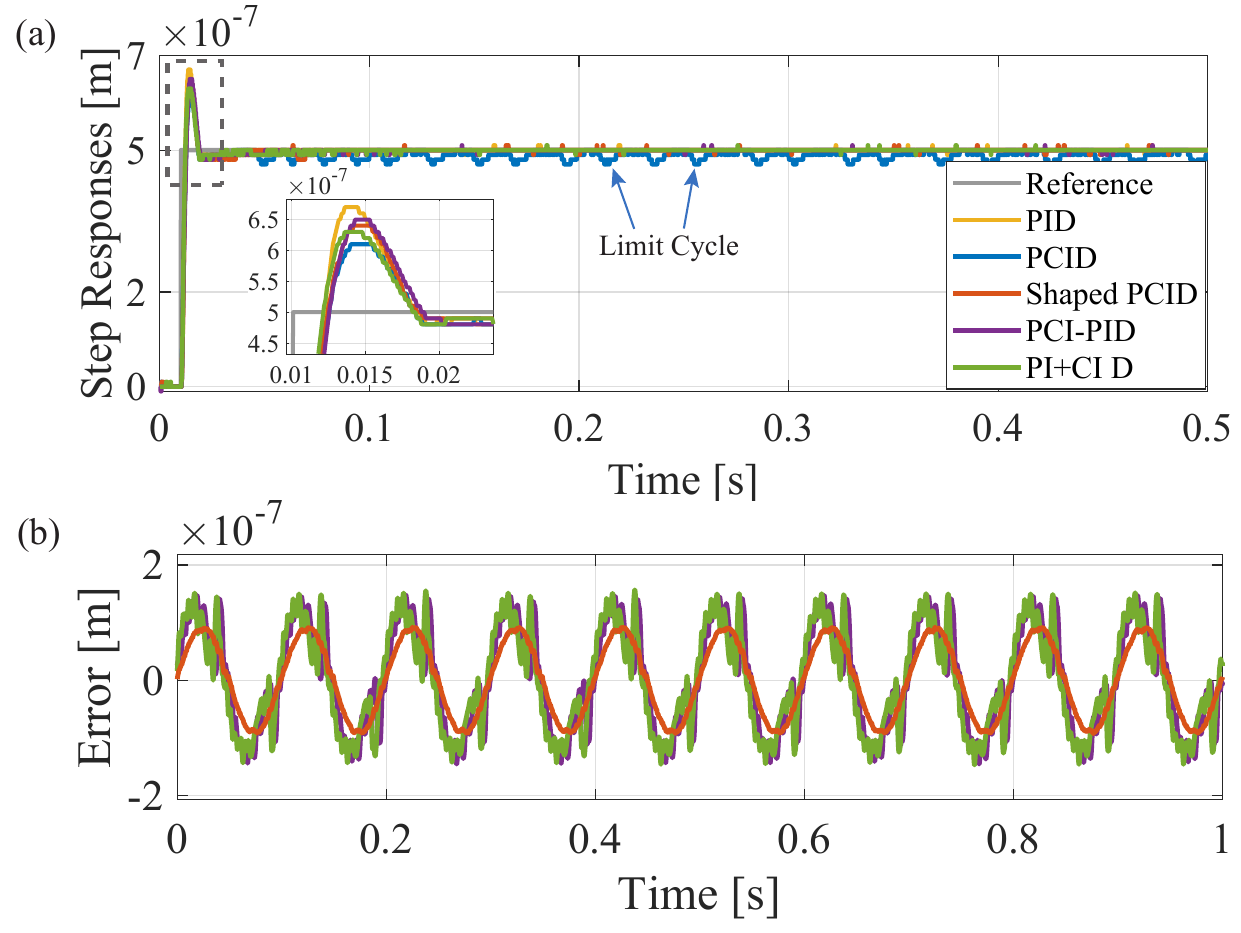}}
    \caption{Experimental measured (a) step responses of PID, PCID, shaped PCID, PCI-PID, PI+CI D control systems, and (b) steady-state errors of the PCI-PID, PI+CI D, and shaped PCID control systems under sinusoidal input \(r(t) = 1 \times 10^{-5} \sin(20\pi t)\) [m].}
	\label{step50_final}
\end{figure}


In summary, the proposed shaped PCID control system improves positioning accuracy and control efficiency compared to both PID and PCID control systems on the precision motion stage. Additionally, it effectively eliminates limit cycles, leading to enhanced overall system performance.



\section{Conclusion}
\label{sec: conclusion}
In conclusion, this paper makes two main contributions. First, it introduces a method for identifying multiple-reset and two-reset regions in sinusoidal-input closed-loop reset systems, providing engineers with a practical tool to evaluate the reliability of Sinusoidal Input Describing Function (SIDF) analysis. The effectiveness and time-saving advantages of this method have been validated through simulations and experimental results across six case studies.

Second, the study introduces a shaped reset control strategy to reduce high-order harmonics. As an illustrative example, the procedure for designing a PID shaping filter in CI-based reset systems is presented. The resulting PID-shaped reset control system reduces high-order harmonics while preserving the benefits of the first-order harmonic compared to the reset control system. Experimental results from precision motion stages highlight three key benefits of the PID-shaped reset system: (1) Improved SIDF analysis accuracy; (2) Enhanced tracking precision, disturbance and noise rejection, and overall control efficiency; and (3) Elimination of limit-cycle issues in the step responses of reset systems.

Future research could explore the application of the shaped reset control system design in Section \ref{sec: New Shaped Reset Systems} to other reset control structures, aiming to investigate further improvements in system performance.

\bibliographystyle{unsrt}
\bibliography{Ref}
\section*{Appendix}
\subsection{The proof of Lemma \ref{lem: piece-wise}}
\label{appendix: proof for Lemma 1}
\begin{proof}
Consider a closed-loop reset control system in Fig. \ref{fig: RCS_d_n_r_n_n} under a sinusoidal reference input \( r(t) = |R| \sin (\omega t) \), and satisfying Assumption \ref{assum: stable}. 

Within each steady-state period \((0, 2\pi/\omega]\), the reset instant \( t_i \) is defined as the time at which the reset-triggered signal \( z_s(t_i) \) reaches zero. Let \( x_i(t) \), \( m_i(t) \), \( z_i(t) \), and \( z_s^i(t) \) represent the state of the reset controller \(\mathcal{C}_r\), the reset output, the reset input, and the reset-triggered signal, during the intervals \((t_{i-1}, t_i]\), where \( i \in \mathbb{Z}^+ \), respectively. This proof presents the piecewise expressions for the steady-state trajectories of the system, following the three steps outlined below.

\textbf{Step 1: Derive the Piecewise Expression for \( x_i(t) \).}


From \eqref{eq: State-Space eq of RC}, the system operates without any reset actions during the time interval \((t_{i-1}, t_{i}]\). At the reset instant \( t_{i} \in J \), the state \( x_{i}(t_{i}) \) undergoes a reset (or jump) to a new state \( x_{i}(t_{i}^+) \), given by
\begin{equation}
\label{x_r(t_i+)}
    x_i(t_{i}^+) = A_\rho x_i(t_{i}).
\end{equation}
The jump in \eqref{x_r(t_i+)} introduces a step input signal \(h_{i}(t)\) into the system, impacting the trajectories during the subsequent time interval \((t_{i}, t_{i+1}]\) \cite{ZHANG2024106063}. The signal $h_{i}(t)$ is given by 
    \begin{equation}
    \label{eq: h_i}
        h_{i}(t) =[x_i(t_{i}^+) - x_i(t_{i})]h (t-t_i) = (A_\rho-I)x_i(t_{i})h(t-t_{i}),
    \end{equation}
where $h(t)$ is a unit step signal given by
\begin{equation}
\label{eq: ht1}
    h(t)={\begin{cases}1,\ &t>0\\0,&t\leq 0\end{cases}},
\end{equation}
with the Fourier transform $ H(\omega) = \mathscr{F}[h(t)] ={(j\omega)}^{-1}$. 

Based on \eqref{eq: State-Space eq of RC} and \eqref{eq: h_i}, the block diagram of the controller \(\mathcal{C}\) for the time interval \((t_i, t_{i+1}]\) is illustrated in Fig. \ref{fig1:xnl_to_x_tf}.
\vspace{-0.25cm}
\begin{figure}[h]
	\centerline{\includegraphics[width=0.35\textwidth]{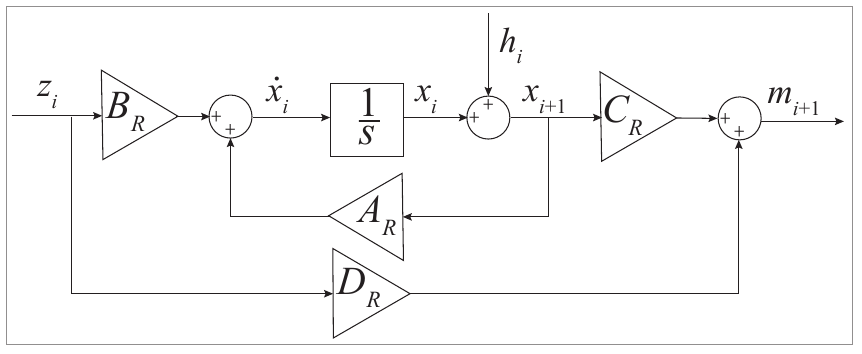}}
	\caption{State-space block diagram of $\mathcal{C}_r$ during the time interval $(t_{i},\ t_{i+1}]$.}
	\label{fig1:xnl_to_x_tf}
\end{figure}

\vspace{-0.25cm}
From Fig. \ref{fig1:xnl_to_x_tf}, the signal \( x_{i+1}(t) \) is derived from two inputs: \( z_i(t) \) and \( h_i(t) \). The respective contributions to \( x_{i+1}(t) \) from \( z_i(t) \) and \( h_i(t) \) are labeled as \( x_{i+1}^z(t) \) and \( x_{i+1}^h(t) \), respectively. 

Under Assumption \ref{assum: stable}, the system is guaranteed to be convergent, ensuring that the trajectories of the system are integrable, as discussed in \cite{ZHANG2024106063}. This integrability implies that the system’s time-domain signals have well-defined Fourier transforms. Let \( Z_i(\omega) \), \( H_i(\omega) \), \( X_i(\omega) \), \( X_{i+1}^z(\omega) \), and \( X_{i+1}^h(\omega) \) represent the Fourier transforms of the signals \( z_i(t) \), \( h_i(t) \), \( x_i(t) \), \( x_{i+1}^z(t) \), and \( x_{i+1}^h(t) \), respectively. 

Since no reset actions occur during the time interval \((t_i, t_{i+1}]\), the superposition law holds. Therefore, $X_{i+1}(\omega)$ is express as
\begin{equation}
\label{Xi+1(w)}
\begin{aligned}
&X_{i+1}(\omega) = X_{i+1}^z(\omega) + X_{i+1}^h(\omega)\\
&= {X_{i+1}^z(\omega)}{Z_{i}(\omega)}^{-1} \cdot Z_{i}(\omega) + {X_{i+1}^h(\omega)}{H_{i}(\omega)}^{-1} \cdot H_{i}(\omega).
\end{aligned}
\end{equation}
Based on Figs. \ref{fig: RCS_d_n_r_n_n} and \ref{fig1:xnl_to_x_tf}, within the closed-loop reset system, when \(h_i(t) = 0\), we have
\begin{equation}
\label{eq: X_i+1_e_trans}
    \begin{aligned}
       {X_{i+1}^z(\omega)}{Z_{i}(\omega)}^{-1} &= (j\omega I-A_R)^{-1}B_R,
    \end{aligned}
\end{equation}
and when \(z_i(t) = 0\), it follows that
\begin{equation}
\label{eq: X_i+1_h_trans}
{X_{i+1}^h(\omega)}{H_{i}(\omega)}^{-1} = \mathcal{S}_{bl}(\omega)(j\omega I-A_R)^{-1}j\omega,
\end{equation}
By combining \eqref{Xi+1(w)}, \eqref{eq: X_i+1_e_trans}, and \eqref{eq: X_i+1_h_trans}, we derive
\begin{equation}
\label{eq:X_i+12}
\begin{aligned}
X_{i+1}(\omega) = (j\omega I-A_R&)^{-1} B_RZ_{i}(\omega) +\mathcal{S}_{bl}(\omega) \cdot\\
&(j\omega I-A_R)^{-1}j\omega H_{i}(\omega).  
\end{aligned}  
\end{equation}
According to \eqref{eq: State-Space eq of RC}, during the reset interval \((t_i, t_{i+1}]\), we obtain \(X_{i}(\omega) = (j\omega I-A_R)^{-1}B_RZ_{i}(\omega)\). Substituting this $X_{i}(\omega)$ into \eqref{eq:X_i+12} yields
\begin{equation}
\label{eq:X_i+13}
X_{i+1}(\omega) = X_{i}(\omega) + \mathcal{S}_{bl}(\omega) (j\omega I-A_R)^{-1}j\omega H_{i}(\omega). 
\end{equation}
From \eqref{eq: h_i}, the Fourier transform of \(h_i(t)\) is given by 
\begin{equation}
\label{eq: Hi_def}
   H_i(\omega) = \mathscr{F}[h_i(t)] = (A_\rho-I)(j\omega)^{-1} e^{-j\omega t_i}x_i(t_{i}).
\end{equation}
Substituting \eqref{eq: Hi_def} into \eqref{eq:X_i+13}, we obtain
\begin{equation}
\label{eq:X_i+14}
\begin{aligned}
X_{i+1}(\omega) = X_{i}(\omega) + \mathcal{T}_{s}(\omega) e^{-j\omega t_i}x_i(t_{i}),
\end{aligned}
\end{equation}
where
\begin{equation}
\mathcal{T}_{s}(\omega) = \mathcal{S}_{bl}(\omega) (j\omega I-A_R)^{-1} (A_\rho-I).
\end{equation}
Conducting the Fourier transforms of equation \eqref{eq:X_i+14}, we obtain:
\begin{equation}
\label{e_i+1}
 x_{i+1}(t) = x_{i}(t) + h_{s}(t-t_i)x_i(t_i),\text{ where }h_{s}(t) = \mathscr{F}^{-1}[\mathcal{T}_{s}(\omega)].
\end{equation}
Till here, the state of the reset controller during the time interval $(t_i,t_{i+1}]$ denoted as $x_{i+1}(t)$ is derived.

\textbf{Step 2: Derive the Piecewise Expression for \( z_i(t) \).}

Similarly to Step 1, from Fig. \ref{fig1:xnl_to_x_tf}, the signal \( m_{i+1}(t) \) is derived from two inputs: \( z_i(t) \) and \( h_i(t) \). The contributions to the output \( m_{i+1}(t) \) from \( z_i(t) \) and \( h_i(t) \) are denoted as \( m_{i+1}^z(t) \) and \( m_{i+1}^h(t) \), respectively. Let \( M_i(\omega) \), \( M_{i+1}^z(\omega) \), and \( M_{i+1}^h(\omega) \) represent the Fourier transforms of the signals \( m_i(t) \), \( m_{i+1}^z(t) \), and \( m_{i+1}^h(t) \), respectively. Using the same calculation process as in Step 1, \( M_{i+1}(\omega) \) is expressed as:
\begin{equation}
\label{Vi+1(w), Xi+1(w)}
\begin{aligned}
\resizebox{1\columnwidth}{!}{$M_{i+1}(\omega) 
= {M_{i+1}^z(\omega)}{Z_{i}(\omega)}^{-1} \cdot Z_{i}(\omega) +{M_{i+1}^h(\omega)}{H_{i}(\omega)}^{-1}\cdot H_{i}(\omega),$}
\end{aligned}
\end{equation}
where
\begin{equation}
\label{eq: Mi+1_z,M_i+1_h}
    \begin{aligned}
       {M_{i+1}^z(\omega)}{Z_{i}(\omega)}^{-1} &=C_R(j\omega I-A_R)^{-1}B_R+D_R= \mathcal{C}_{l}(\omega),\\
       {M_{i+1}^h(\omega)}{H_{i}(\omega)}^{-1} &= \mathcal{S}_{bl}(\omega)C_R(j\omega I-A_R)^{-1}j\omega.       
    \end{aligned}
\end{equation}
Substituting \eqref{eq: Mi+1_z,M_i+1_h} into \eqref{Vi+1(w), Xi+1(w)}, $M_{i+1}(\omega)$ is simplified to
\begin{equation}
\label{eq: Mi+1_1}
\begin{aligned}
\resizebox{1\columnwidth}{!}{$M_{i+1}(\omega) = \mathcal{C}_{l}(\omega)Z_{i}(\omega) + \mathcal{S}_{bl}(\omega)C_R(j\omega I-A_R)^{-1}j\omega H_{i}(\omega).$}
\end{aligned}
\end{equation}
From \eqref{eq: State-Space eq of RC}, during the reset interval \((t_i, t_{i+1}]\), we  have
\begin{equation}
\label{eq:Mi0}
M_{i}(\omega) = \mathcal{C}_{l}(\omega)Z_{i}(\omega).    
\end{equation}
Substituting \eqref{eq: Hi_def} and \eqref{eq:Mi0} into \eqref{eq: Mi+1_1}, $M_{i+1}(\omega)$ is given by
\begin{equation}
\label{eq: Mi+1_2}
\begin{aligned}
M_{i+1}(\omega) = M_{i}(\omega) + C_R \mathcal{T}_{s}(\omega)e^{-j\omega t_i}x_i(t_{i}).
\end{aligned}
\end{equation}
From Fig. \ref{fig: RCS_d_n_r_n_n}, in the closed-loop reset system, the following relation holds:
\begin{equation}
\label{eq: V(w),E(w)}
\begin{aligned}
Z_{i}(\omega)& = R(\omega) - (A(\omega)+M_{i}(\omega))\mathcal{C}_\sigma(\omega),\\
Z_{i+1}(\omega) &= R(\omega) - (A(\omega)+M_{i+1}(\omega))\mathcal{C}_\sigma(\omega),
\end{aligned}   
\end{equation}
where $A(\omega) = \mathscr{F}[a(t)] $ and
\begin{equation}
\mathcal{C}_\sigma(\omega) =\mathcal{C}_{3}(\omega)\mathcal{P}(\omega)\mathcal{C}_{4}(\omega)\mathcal{C}_{1}(\omega).
\end{equation}
From \eqref{eq: V(w),E(w)}, the following equations are derived
\begin{equation}
\label{eq: V(w),E(w)2}
\begin{aligned}
M_{i}(\omega)&= [{R(\omega)-Z_{i}(\omega)}]\cdot{\mathcal{C}_\sigma(\omega)}^{-1}-A(\omega),\\
M_{i+1}(\omega)&= [{R(\omega)-Z_{i+1}(\omega)}]\cdot{\mathcal{C}_\sigma(\omega)}^{-1}-A(\omega).
\end{aligned}   
\end{equation}
Substituting $M_{i+1}(\omega)$ and $M_{i}(\omega)$ from \eqref{eq: V(w),E(w)2} into \eqref{eq: Mi+1_2}, $Z_{i+1}(\omega)$ is obtained as
\begin{equation}
\label{E_i+1(w)}
Z_{i+1}(\omega) = Z_{i}(\omega) - \mathcal{T}_{\alpha}(\omega)e^{-j\omega t_i}x_i(t_{i}),
\end{equation}
where 
\begin{equation}
\label{eq: T_alpha_def}
\begin{aligned}
\mathcal{T}_{\alpha}(\omega) &= \mathcal{C}_{\sigma}(\omega)C_R\mathcal{T}_{s}(\omega).
\end{aligned}
\end{equation}
Conducting the Fourier transforms of equation \eqref{E_i+1(w)}, we obtain:
\begin{equation}
\label{e_i+1}
 z_{i+1}(t) = z_{i}(t) - h_{\alpha}(t-t_i)x_i(t_i),\text{ where }h_{\alpha}(t) = \mathscr{F}^{-1}[\mathcal{T}_{\alpha}(\omega)].
\end{equation}
Till here, the input of the reset controller during the time interval $(t_i,t_{i+1}]$ denoted as $z_{i+1}(t)$ is derived.

\textbf{Step 3: Derive the Piecewise Expression for \( z_s^i(t) \).}

%

During the reset intervals \((t_{i-1}, t_{i}]\) and \((t_i, t_{i+1}]\), no reset action takes place. Let \( Z_s^i(\omega) \) denotes the Fourier transforms of \( z_s^i(t) \). Therefore, the following relationship holds:
\begin{equation}
\label{eq:z_zi,z_si+1}
\begin{aligned}
Z_s^{i}(\omega) = \mathcal{C}_s(\omega)Z_{i}(\omega),\text{ and } Z_s^{i+1}(\omega) = \mathcal{C}_s(\omega)Z_{i+1}(\omega). 
\end{aligned} 
\end{equation}
Substituting \eqref{eq:z_zi,z_si+1} into \eqref{E_i+1(w)}, we obtain
\begin{equation}
\label{eq:Z_si+1}
Z_s^{i+1}(\omega) = Z_s^{i}(\omega) - \mathcal{C}_s(\omega)\mathcal{T}_{\alpha}(\omega)e^{-j\omega t_i}x_i(t_{i}).
\end{equation}
Conducting the Fourier transforms of equation \eqref{eq:Z_si+1}, we obtain:
\begin{equation}
\label{e_i+1}
 z_s^{i+1}(t) = z_s^{i}(t) - h_{\beta}(t-t_i)x_i(t_i),
\end{equation}
where
\begin{equation}
\label{h_beta(t)}
     h_{\beta}(t) = \mathscr{F}^{-1}[\mathcal{C}_s(\omega)\mathcal{T}_{\alpha}(\omega)].
\end{equation}
Till here, the expression of the reset triggered signal during the time interval $(t_i,t_{i+1}]$ denoted as $z_s^{i+1}(t)$ is derived. We conclude the proof.
\end{proof}

\vspace{-0.5cm}
\subsection{The proof of Theorem \ref{thm: Delta}}
\label{appendix: proof for Theorem 1}
\begin{proof}
Consider the reset control system shown in Fig. \ref{fig: RCS_d_n_r_n_n} with a sinusoidal reference input \( r(t) = |R|\sin(\omega t) \) and satisfies Assumptions \ref{assum: stable} and \ref{assum: t1}. This proof derives the multiple-reset conditions in the sinusoidal-input reset system. It is organized into four steps as follows.

\textbf{Step 1: Derive the First Reset Instant \( t_1 \) Within One Steady-State Cycle.}

Under Assumption \ref{assum: t1}, the state and reset-triggered signal of the reset system during the interval \((0, t_1]\), denoted as \(x_1(t)\) and \(z_s^1(t)\), are equivalent to those of its BLS, denoted as \(x_{bl}(t)\) and \(z_{bl}(t)\), respectively, as expressed by:
\begin{equation}
\label{eq: x1,e1,es}
\begin{aligned}
    x_1(t) &= x_{bl}(t) = |R \Theta_{bl}(\omega)|\sin(\omega t + \angle \Theta_{bl}(\omega)), \\  
    z_s^1(t)  & = z_{bl}(t) = |R\mathcal{S}_{ls}(\omega)|\sin(\omega t +\angle \mathcal{S}_{ls}(\omega) ),  
\end{aligned}
\end{equation}
where $\angle \Theta_{bl}(\omega) \in (-\pi, \ \pi]$ and $\angle \mathcal{S}_{ls}(\omega) \in (-\pi, \ \pi]$. Functions $\Theta_{bl}(\omega)$ and $\mathcal{S}_{ls}(\omega)$ are given in \eqref{eq: ht}.

From Assumption \ref{assum: t1} and \eqref{eq: x1,e1,es}, the first reset instant denoted as $t_1$ within one steady-state cycle, which corresponds to the first zero-crossing point of the reset-triggered signal \( z_s^1(t)\), is expressed as:
	\begin{equation}
		\label{eq: t1}
		t_{1} = 
  \begin{cases}
     {(\pi-\angle\mathcal{S}_{ls}(\omega))}/{\omega}, & \text{ if }\angle \mathcal{S}_{ls}(\omega) \in (0, \ \pi],\\
     {(-\angle\mathcal{S}_{ls}(\omega))}/{\omega}, & \text{ if }\angle \mathcal{S}_{ls}(\omega) \in (-\pi, \ 0].    
  \end{cases}
\end{equation}
From \eqref{eq: t1}, we have \( t_1 \leq \pi/\omega \).

\textbf{Step 2: Draw the Conclusion that Reset Instants Occurring \(\pi/\omega\)-Periodically.}

Under Assumption \ref{assum: stable}, the reset-triggered signal \( z_s(t) \) in the sinusoidal-input reset system can be represented as an infinite series of harmonics \cite{pavlov2006uniform}, denoted by \( z_{sn}(t) \), and is given by
\begin{equation}
\label{eq:x_r(t)}
\begin{aligned}
z_s(t) &= \sum\nolimits_{n=1}^{\infty}z_{sn}(t) = \sum\nolimits_{n=1}^{\infty} |Z_{sn}|\sin(n\omega t + \angle Z_{sn}),
\end{aligned}
\end{equation}
where \( |Z_{sn}| \) denotes the magnitude and \( \angle Z_{sn} \) represents the phase of each harmonic component \( z_{sn}(t) \).

From \eqref{eq:x_r(t)}, we obtain
\begin{equation}
\label{eq:x_r(t)+-pi2}
 z_s(t_i) = - z_s(t_i\pm\pi/\omega) = 0. 
\end{equation}
From \eqref{eq:x_r(t)+-pi2}, the reset instant \( t_i \), where \( z_s(t_i) = 0 \), occurs periodically with a period of \( \pi/\omega \).

\textbf{Step 3: Establish the Condition for Multiple-Reset Systems: The Reset Triggered Signal \( z_s^2(t) \) Must Exhibit at Least One Zero-Crossing Within the Interval \( (t_1, \pi/\omega) \).}

From \eqref{eq:x_r(t)+-pi2}, within a steady-state period \((0, 2\pi/\omega]\), we obtain two conclusions:
\begin{enumerate}
    \item At the time instant $t_1$ and $t_1+\pi/\omega$, we have $z_s(t_1) = z_s(t_1 + \pi/\omega) = 0$.
\item Since \( t_1 \) represents the first reset instant within a steady-state cycle \((0, 2\pi/\omega]\), there is no zero-crossings of \( z_s(t) \) in the both the time intervals \((0, t_1)\) and \((\pi/\omega, t_1 + \pi/\omega)\).
\end{enumerate}

From these two conclusions, Fig. \ref{multiple_reset_condition} shows the green area that have no reset actions within a steady-state period \((0, 2\pi/\omega]\).
\vspace{-0.35cm}
\begin{figure}[htp]
    \centering
    \centerline{\includegraphics[width=0.75\columnwidth]{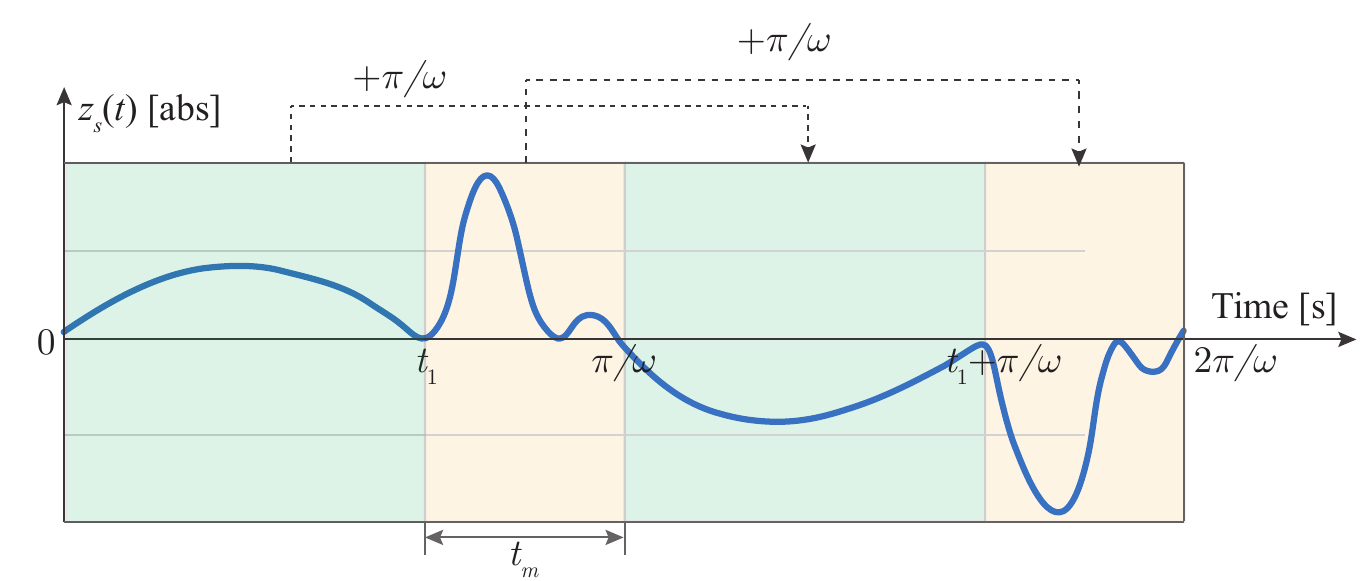}}
    \caption{Plot of the steady-state reset-triggered signal $z_s(t)$, with regions of same color indicating the same number of zero-crossings and opposite-sign trajectories. Green areas indicate no zero-crossings, while yellow areas show regions with zero-crossings.}
	\label{multiple_reset_condition}
\end{figure}
\vspace{-0.25cm}

A system is classified as a multiple-reset system if it exhibits more than two zero-crossings per \( 2\pi/\omega \) steady-state cycle in response to a sinusoidal reference input \( r(t) = |R|\sin(\omega t) \). If the reset-triggered signal \( z_s(t) \) has no zero-crossings within the interval \( (t_1, \pi/\omega) \), it will lack zero-crossings within \( (t_1 + \pi/\omega, 2\pi/\omega) \). This results in exactly two zero-crossings at \( t_1 \) and \( t_1 + \pi/\omega \) over one steady-state cycle, \((0, 2\pi/\omega] \). 

Therefore, a system exhibits multiple-reset behavior if the reset-triggered signal \( e_s(t) \) has at least one zero-crossing within \( (t_1, \pi/\omega) \), as illustrated in the yellow-shaded area of Fig. \ref{multiple_reset_condition}.

From Lemma \ref{lem: piece-wise}, \( z_s(t) \) can be broken down into piecewise components \( z_s^i(t) \) over intervals \( (t_i, t_{i+1}] \). Thus, for \( z_s(t) \) to have at least one zero-crossing within \( (t_1, \pi/\omega) \), the second piece \( z_s^2(t) \) must include at least one zero-crossing within the interval \( t \in (t_1, \pi/\omega) \).

\textbf{Step 4: Formulate the Multiple-Reset Condition.}

From \eqref{eq: xi, zi, zsi}, the reset-triggered signal \(z_s^2(t) \) during the time interval \((t_1, t_2]\) can be expressed as:
\begin{equation}
\label{es_2_true}  
z_s^2(t) = z_s^1(t) - h_{\beta}(t-t_1)x_1(t_1), \text{ for } t\in(t_1, t_2].
\end{equation}
From \eqref{eq: x1,e1,es} and \eqref{eq: t1}, $x_{1}(t_1)$ is given by
\begin{equation}
\label{eq: xbl_t1}
\begin{aligned}
x_{1}(t_1) &= |R \Theta_{bl}(\omega)|\sin(\omega t_1 + \angle \Theta_{bl}(\omega))\\
&=
\begin{cases}
|R|\cdot \Theta_s(\omega) , & \text{ if }\angle \mathcal{S}_{ls}(\omega) \in (0, \ \pi],\\
-|R|\cdot \Theta_s(\omega), & \text{ if }\angle \mathcal{S}_{ls}(\omega) \in (-\pi, \ 0],
\end{cases}
\end{aligned}
\end{equation}
where
\begin{equation}
\Theta_s(\omega) = |\Theta_{bl}(\omega)|\sin(\angle \mathcal{S}_{ls}(\omega) - \angle \Theta_{bl}(\omega)).
\end{equation}
By defining \( t = t_\delta + t_1 \) and substituting \( x_{1}(t_1) \) from \eqref{eq: xbl_t1} into \eqref{es_2_true}, along with $z_s^1(t)$ defined from \eqref{eq: x1,e1,es}, we obtain:
\begin{equation}
\label{eq: es2=es1+er2_2}
\begin{aligned}
z_s^2(t_\delta+t_{1})  = 
\begin{cases}
 -|R|\Delta(t_\delta) , & \text{ if }\angle \mathcal{S}_{ls}(\omega) \in (0, \ \pi],\\   
|R|\Delta(t_\delta), & \text{ if }\angle \mathcal{S}_{ls}(\omega) \in (-\pi, \ 0], 
\end{cases}
\end{aligned}
\end{equation}
where 
\begin{equation}    
\Delta(t_\delta) = |\mathcal{S}_{ls}(\omega)|\sin(\omega t_\delta)+ h_{\beta}(t_\delta) \Theta_s(\omega ).
\end{equation}
The multiple-reset condition requires that \( z_s^2(t) \) has at least one zero-crossing within the time interval \( (t_1, \pi/\omega) \). Using the relation \( t = t_\delta + t_1 \) and from \eqref{eq: es2=es1+er2_2}, this condition is transformed to: there exists a time interval \( t_\delta \in (0, \pi/\omega - t_1) \) such that \( z_s^2(t_\delta + t_1) \) has at least one zero-crossing.

From \eqref{eq: t1}, the value of $\pi/\omega-t_1$ is given by 
\begin{equation}
		\label{eq: pi/w-t1}
\pi/\omega-t_1 = 
  \begin{cases}
     {(\angle\mathcal{S}_{ls}(\omega))}/{\omega}, & \text{ if }\angle \mathcal{S}_{ls}(\omega) \in (0, \ \pi],\\
     {(\pi+\angle\mathcal{S}_{ls}(\omega))}/{\omega}, & \text{ if }\angle \mathcal{S}_{ls}(\omega) \in (-\pi, \ 0].    
  \end{cases}
\end{equation}
From \eqref{eq: pi/w-t1}, $\pi/\omega-t_1$ can be expressed as
\begin{equation}
\label{eq: tm}
    t_m = \pi/\omega-t_1 = \angle\mathcal{S}_{ls}(\omega)/{\omega} + \pi/\omega \cdot \text{sign}(\mathcal{S}_{ls}(\omega),
\end{equation}
where
\begin{equation}
\text{sign}(x) = \begin{cases} 
      0, & \text{if } x > 0, \\
      1, & \text{if } x \leq 0.
   \end{cases}
\end{equation}
Since a zero-crossing is independent of amplitude, 
the multiple-reset condition is simplified to verifying the existence of a time interval \( t_\delta \in (0, t_m) \) such that \( \Delta(t_\delta) = 0 \). This completes the proof.

\end{proof}

\vspace{-1cm}
\subsection{The proof of Lemma \ref{lem: stair_step}}
\begin{proof}
\label{appendix: proof for Lemma 2}
Consider a closed-loop reset control system as shown in Fig. \ref{fig: RCS_d_n_r_n_n}, with a sinusoidal reference input \( r(t) = |R| \sin (\omega t) \) and satisfying Assumptions \ref{assum: stable} and \ref{assum: t1}. This proof demonstrates that the steady-state reset-triggered signal \( z_s(t) \) is composed of a base-linear component \( z_{bl}(t) \) and a nonlinear component \( z_{nl}(t) \), where \( z_{nl}(t) \) is obtained by filtering a stair-step signal $d_s(t)$ through an LTI transfer function. The proof is organized into three steps.

\textbf{Step 1: Prove that Reset Actions Introduce Square Waves into Systems.}

The state \( x_c(t) \) of the reset controller \( \mathcal{C}_r \) is nonlinear and can be represented as the sum of its harmonics \cite{pavlov2006uniform}, expressed as
\begin{equation}
\label{eq:xc(t)_def}
    x_c(t) = \sum\nolimits_{n=1}^{\infty} x_{cn}(t) = |X_{cn}| \sin(n\omega t + \angle X_{cn}),
\end{equation}
where $|X_{cn}|$ and $\angle X_{cn}$ represent the magnitude and phase of each harmonic $x_{cn}(t)$ in $ x_c(t)$.

From \eqref{eq:xc(t)_def}, the following relation holds
\begin{equation}
\label{eq:xc(ti)}
x_c(t_i) = -x_c(t_i\pm\pi/\omega).  
\end{equation}


Based on \eqref{eq:x_r(t)+-pi2}, the reset instant \( t_i \) occurs \( \pi/\omega \)-periodically. At each reset instant \( t_i \), according to \eqref{eq: h_i}, the state \( x_c(t_i) \) undergoes a jump to \( A_\rho x_c(t_i) \), generating a step input defined by \( h_{i}(t) = (A_\rho - I)x_c(t_{i}) \, h(t - t_{i}) \). 

Then, based on \eqref{eq:xc(ti)}, a step input with an opposite sign \( h'_{i}(t) = - h_{i}(t) = -(A_\rho - I)x_c(t_{i}) \, h(t - t_{i} \pm \pi/\omega) \) is introduced at the time instant \( t_i + \pi/\omega \). Signals $h_i(t)$ and $- h_{i}(t)$ together produce a square wave signal over each steady-state cycle, beginning at \( t_i \) with an amplitude of \( (A_\rho - I)x_i(t_{i}) \) and a period of \( 2\pi/\omega \), as illustrated in Fig.\ref{square_wave_t1_t2_t_mu}.

\vspace{-0.3cm}
\begin{figure}[htp]
    \centering
    \centerline{\includegraphics[width=0.75\columnwidth]{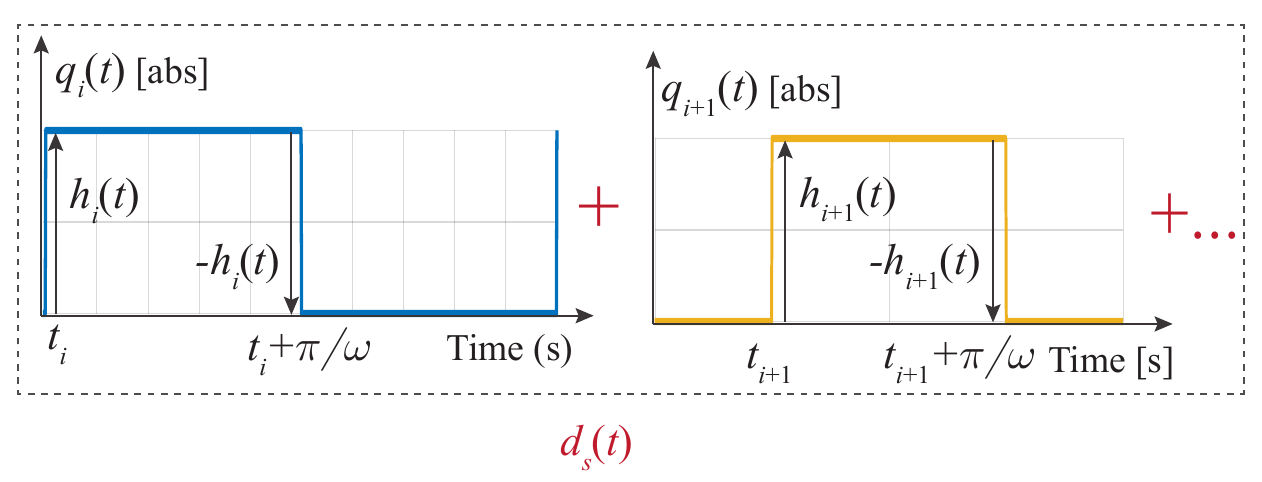}}
    \caption{Plots of signals $q_i(t)$ in \eqref{eq: qi_def} and $d_s(t)$ in \eqref{eq: ds(t)_def0}.}
	\label{square_wave_t1_t2_t_mu}
\end{figure}
\vspace{-0.3cm}

\textbf{Step 2: Formulate the Square Waves.}

The square wave introduced at the time instants $t_i$ and $t_i+\pi/\omega$ is expressed as 
\begin{equation}
\label{eq: qi_def}
 q_i(t) = (A_\rho-I)x(t_i)q(t-t_i),   
\end{equation}
where \( q(t) \) is a square wave with an amplitude of 1 and a period of \( 2\pi/\omega \), defined as:
\begin{equation}
\label{eq:q(t)}
q(t) = \sum\nolimits_{n=1}^{\infty} 2\cdot{\sin(n\omega t)}/{n\pi},\ n = 2k+1, k\in\mathbb{N}.
\end{equation}
From \eqref{eq: qi_def} and \eqref{eq:q(t)}, $q_i(t)$ is expressed as
\begin{equation}
\label{eq: qi_def2}
    \begin{aligned}
 q_i(t) &=  \sum\nolimits_{n=1}^{\infty} q_i^n(t),
    \end{aligned}
\end{equation}
where 
\begin{equation}
\label{eq: qi_def22}
q_i^n(t) =  {2(A_\rho-I)x(t_i)\sin(n\omega (t-t_i))}/({n\pi}).
\end{equation}

\textbf{Step 3: Illustrate that Square Waves $q_i(t)$ Combine to Form a Stair-Step Signal $d_s(t)$, Contributing to the Reset-Triggered Signal $z_s(t)$.}

At each reset instant \( t_i \) within the half-cycle \((0, \pi/\omega]\), a square wave \( q_i(t) \) is introduced. Let the number of reset instants within each half-cycle \((0, \pi/\omega]\) be denoted by \(\mu\). From \eqref{eq: qi_def2} and \eqref{eq: qi_def22}, a stair-step signal \( d_s(t) \) is generated within one \( 2\pi/\omega \) period. This signal is illustrated in Fig. \ref{square_wave_t1_t2_t_mu} and is expressed as:
\begin{equation}
\label{eq: ds(t)_def0}
    d_s(t) = \sum\nolimits_{i=1}^{i=\mu} q_i(t) = \sum\nolimits_{i=1}^{i=\mu}\sum\nolimits_{n=1}^{\infty} q_i^n(t).
\end{equation}
From \eqref{eq: ds(t)_def0}, $d_s(t)$ can be written as
\begin{equation}
\label{eq: ds(t)_def1}
    d_s(t) = \sum\nolimits_{n=1}^{\infty}\sum\nolimits_{i=1}^{i=\mu} q_i^n(t).
\end{equation}
Define $d_s^n(t)$ as the $n$th harmonic of $d_s(t)$, from \eqref{eq: qi_def22} and \eqref{eq: ds(t)_def1}, $d_s(t)$ is expressed as
\begin{equation}
\label{eq: ds(t)_def2}
    \begin{aligned}
     d_s(t) &= \sum\nolimits_{n=1}^{\infty} d_s^n(t),\\
     d_s^n(t) &=  {2(A_\rho-I)}/({n\pi})\cdot\sum\nolimits_{i=1}^{i=\mu} x(t_i)\sin(n\omega (t-t_i)),
    \end{aligned}
\end{equation}
with their Fourier transforms given by
\begin{equation}
\label{eq: Ds(w),Ds_n(w)}
   \begin{aligned}
      D_s(\omega) &=  \sum\nolimits_{n=1}^{\infty} D_s^n(\omega),\\
      D_s^n(\omega)  &= {2(A_\rho-I)}/({n\pi})\cdot\sum\nolimits_{i=1}^{i=\mu}\mathscr{F}[x(t_i)\sin(n\omega (t-t_i))].
   \end{aligned} 
\end{equation}

Under Assumption \ref{assum: t1}, the reset-triggered signal \( z_s(t) \) initially follows its base-linear trajectory \( z_{bl}(t) \) within the interval \((0, t_1)\), as defined in \eqref{eq:zs,znl,zbl}. At time \( t_1 \), reset actions introduce a stair-step signal \( d_s(t) \) into the system. By replacing the signal \( h_i(t) \) (whose Fourier transform is \( H_i(\omega) = 1/(j\omega) \)) with the stair-step signal \( d_s(t) \) (whose Fourier transform is \( D_s(\omega)\)) in Fig. \ref{fig1:xnl_to_x_tf}, and following the derivation process outlined in \ref{appendix: proof for Lemma 1}, the nonlinear component \( z_{nl}(t) \) is derived. Finally, \( z_{bl}(t) \) and \( z_{nl}(t) \) combine to form \( z_s(t) \), as expressed in \eqref{eq:zs,znl,zbl}. This concludes the proof.
\end{proof}

\vspace{-0.5cm}
\subsection{The proof of Theorem \ref{thm: beta_closed_loop_reset}}
\begin{proof}
\label{appendix: proof for Theorem 2}
Consider a closed-loop reset control system as illustrated in Fig. \ref{fig: RCS_d_n_r_n_n}, with a sinusoidal reference input \( r(t) = |R| \sin (\omega t) \), satisfying Assumptions \ref{assum: stable} and \ref{assum: t1}. This proof derives the magnitude ratio of the higher-order harmonics (for \( n > 1 \)) relative to the first-order harmonic (for \( n = 1 \)) in the nonlinear component \( z_{nl}(t) \) as defined in \eqref{eq:zs,znl,zbl}.

From \eqref{eq:zs,znl,zbl} and \eqref{eq: z_{nl}^n,z_nl}, 
the signal \( z_{nl}^n(t) \), representing the \( n \)th harmonic component of \( z_{nl}(t) \), is given by:
\begin{equation}
\label{eq:z_nl^n}
z_{nl}^n(t) = -\mathscr{F}^{-1}[\mathcal{C}_s(n\omega)\mathcal{T}_\beta(n\omega)D_s^n(\omega)].
\end{equation}
From \eqref{eq:z_nl^n}, the Fourier transform of $z_{nl}^n(t)$ is given by
\begin{equation}
\label{eq:z_nl^n(w)}
Z_{nl}^n(\omega) = -\mathcal{C}_s(n\omega)\mathcal{T}_\beta(n\omega)D_s^n(\omega).   
\end{equation}
From \eqref{eq: Ds(w),Ds_n(w)} and \eqref{eq:z_nl^n(w)}, we obtain:
\begin{equation}
\label{eq: beta_n_derive}
\begin{aligned}
\beta_n(\omega) &= \frac{|Z_{nl}^n(\omega)|}{|Z_{nl}^1(\omega)|}
= \frac{|\mathcal{C}_s(n\omega)\mathcal{T}_\beta(n\omega)|}{n|\mathcal{C}_s(\omega)\mathcal{T}_\beta(\omega)|}.
\end{aligned}
\end{equation}
Here, the proof is concluded.
\end{proof}

\subsection{The proof of Remark \ref{lem: angle cs_bw>0}}
\label{appendix: proof for lemma Cs_bw>0}
\begin{proof}
Given the matrices for the generalized CI are defined as \((A_R, B_R, C_R, D_R) = (0, 1, 1, 0)\) and \(A_\rho = \gamma \in (-1, 1)\). From \eqref{eq: Ln}, the SIDF $\mathcal{C}_r^1(\omega)$ for the generalized CI is given by
\begin{equation}
	\label{eq: Hn_CI} 
    \mathcal{C}_r^1(\omega) = (j\omega)^{-1}(\Theta_{CI}(\omega)+1),
\end{equation}
where
\begin{equation}
\label{eq: Hn_CI2}
    \Theta_{CI}(\omega) = \frac{4j(1-\gamma)e^{j\angle \mathcal{C}_s(\omega)}(\cos(\angle \mathcal{C}_s(\omega)))}{\pi(1+\gamma)}.
\end{equation}
From \eqref{eq: Hn_CI} and \eqref{eq: Hn_CI2}, under the condition \(\gamma \in (-1,1)\), \(\angle \mathcal{C}_s(\omega)\) and \(\angle \mathcal{C}_1(\omega)\) are positively correlated. This means that the PI shaping filter in \eqref{eq: cs_pi} with a negative phase, will introduce a phase lag to the first-order harmonic \(\angle \mathcal{C}_1(\omega)\). 

Conversely, a positive phase \(\angle \mathcal{C}_s(\omega)\) introduces a phase lead to \(\angle \mathcal{C}_1(\omega)\), counteracting the phase lag induced by the PI. To achieve this, a derivative element \(\frac{s/\omega_\beta+1}{s/\omega_\eta +1}\), which satisfies \(\angle \mathcal{C}_s(\omega_{BW}) > 0\), is incorporated into the PI shaping filter \(\mathcal{C}_s\) in \eqref{eq: cs_pi} to provide the desired phase lead. However, because the derivative element amplifies high-frequency components, a low-pass filter, \(\frac{1}{s/\omega_\psi+1}\), is added to suppress these components and prevent \(z_s(t)\) from becoming overly sensitive to high-frequency noise.
Here, the proof is concluded.
\end{proof}

\end{document}